\documentclass[letterpaper,journal]{IEEEtran} 

\usepackage{cite}
\usepackage{amsmath,amsfonts,amscd,xspace}
\usepackage{epsfig} 
\usepackage{mathptmx} 
\DeclareMathAlphabet{\mathcal}{OMS}{cmsy}{m}{n}
\usepackage{helvet}
\usepackage{courier}
\usepackage{type1cm}  
\usepackage{xcolor}
\usepackage{bbm}
\usepackage{dblfloatfix}
\usepackage{bm}
\usepackage{array}
\usepackage{mathtools}
\usepackage{caption}
\usepackage{wrapfig}
\usepackage[utf8]{inputenc}

\usepackage[ruled]{algorithm2e}
\usepackage{algpseudocode}

\usepackage{algorithm}

\usepackage[T1]{fontenc}
\usepackage{url}
\usepackage{float}
\usepackage{booktabs}

\usepackage{nicefrac}
\usepackage{microtype}   
\usepackage{graphicx}
\usepackage{multirow}
\graphicspath{{Figs/}}

\usepackage{amsthm}

\newtheorem{lemma}{Lemma}
\newtheorem{theorem}{Theorem}
\numberwithin{corollary}{section}
\theoremstyle{definition}
\newtheorem{definition}{Definition}
\numberwithin{definition}{section}

\newtheorem{assumption}{Assumption}

\newcommand{\ouralg}{{\texttt{LAC}}\xspace}
\newcommand{\dcl}{{\texttt{DCL}}\xspace}

\newcommand{\nn}{n}
\newcommand{\nm}{m}
\newcommand{\nk}{k}

\newcommand{\nt}{T}

\makeatletter
\newenvironment{myprocedure}[1][htb]{%
    \renewcommand{\ALG@name}{Procedure}% Update algorithm name
   \begin{algorithm}[#1]%
  }{\end{algorithm}}
\makeatother

\usepackage{hyperref}
\usepackage{cleveref}

\begin{document}
\title{Learning-Augmented Control: Adaptively  Confidence Learning for Competitive MPC}
\author{Tongxin Li, \IEEEmembership{Member, IEEE}
% \thanks{
% This paragraph of the first footnote will contain the date on 
% which you submitted your paper for review. 
% It will also contain support 
% information, including sponsor and financial support acknowledgment. For 
% example, ``This work was supported in part by the U.S. Department of 
% Commerce under Grant 123456.'' 
% }
\thanks{Tongxin Li is with the School of Data Science, The Chinese University of Hong Kong, Shenzhen 518172 China (e-mail: \texttt{litongxin@cuhk.edu.cn}). }}

\maketitle

\begin{abstract}
We introduce \textsc{L}earning-\textsc{A}ugmented \textsc{C}ontrol (\ouralg), an approach that integrates untrusted machine learning predictions into the control of constrained, nonlinear dynamical systems. \ouralg is designed to achieve the ``best-of-both-worlds'' guarantees, i.e, near-optimal performance when predictions are accurate, and robust, safe performance when they are not.  The core of our approach is a delayed confidence learning procedure that optimizes a confidence parameter online, adaptively balancing between ML and nominal predictions. We establish formal competitive ratio bounds for general nonlinear systems under standard MPC regularity assumptions. For the linear quadratic case, we derive a competitive ratio bound that is provably tight, thereby characterizing the fundamental limits of this learning-augmented approach.  The effectiveness of \ouralg is demonstrated in numerical studies, where it maintains stability and outperforms standard methods under adversarial prediction errors.
\end{abstract}

% \begin{IEEEkeywords}
% Enter key words or phrases in alphabetical order, separated by commas. Using the IEEE Thesaurus can help you find the best standardized keywords to fit your article. Use the thesaurus access request form for free access to the IEEE Thesaurus: \underline{https://www.ieee.org/publications/services/thesaurus-acce}\\
% \underline{ss-page.com.}
% \end{IEEEkeywords}

\section{Introduction}
\label{sec:introduction}

\IEEEPARstart{P}{redictive}  information is fundamental to controlling dynamical systems, enabling methods like Model Predictive Control (MPC) to optimize actions based on anticipated future behavior~\cite{garcia1989model}.
While MPC has widespread applications across numerous engineering domains, its performance critically depends on the quality of these predictions. In the current era of artificial intelligence, these predictions increasingly originate from sophisticated yet inherently uncertain sources, particularly off-the-shelf Machine Learning (ML) models. For instance, ML-generated predictions are utilized in data-driven MPC (see~\cite{rosolia2018data,hewing2020learning} and~\cite{bold2024data}) and model-based reinforcement learning~\cite{park2023predictable}, while other forms of auxiliary predictive advice like human feedback are employed in applications including mobile robotics~\cite{chipalkatty2013less} and fair motion planing~\cite{villa2025fair}.
% energy management~\cite{li2024out}.  
A major limitation, however, is that these ML-based and other predictive sources often lack formal guarantees on their accuracy. This uncertainty introduces significant risks, especially in dynamic, time-varying systems, potentially leading to degraded performance or even instability. 
 
The performance and robustness of MPC under uncertain inputs have been extensively studied in the literature, often focusing on sensitivity analysis within given uncertainty bounds; see,
e.g.,~\cite{cannon2005optimizing,quevedo2010input,fleming2014robust,lorenzen2019robust,hanema2020heterogeneously,shin2021controllability,lin2022bounded}, and~\cite{shin2022exponential}.  In many real-world applications, advanced complex, data-driven predictors often produce outputs whose accuracy is difficult to verify a priori, and prediction errors may only appear occasionally. 
While theoretically sound, existing approaches become overly conservative because simple uncertainty models are inadequate for capturing rare but significant prediction errors. 
Consequently, this creates a critical dilemma for control design: (1) leverage potentially superior but uncertain predictions, risking poor outcomes if they are wrong, or (2) rely on safe but suboptimal nominal predictions. The latter approach is frequently adopted in practice to ensure system safety and predictability, albeit at the cost of performance.

Motivated by this practical challenge, we seek to develop an MPC framework that escapes this dichotomy by providing the best of both worlds: (1).
Exploit good predictions \textbf{consistently} to achieve near-optimal performance; (2).
Ensure performance remains close to that of using only nominal predictions \textbf{robustly} when prediction quality degrades significantly (often in the worst-case). 
\begin{figure}[t]
    \centering    \includegraphics[width=0.95\linewidth]{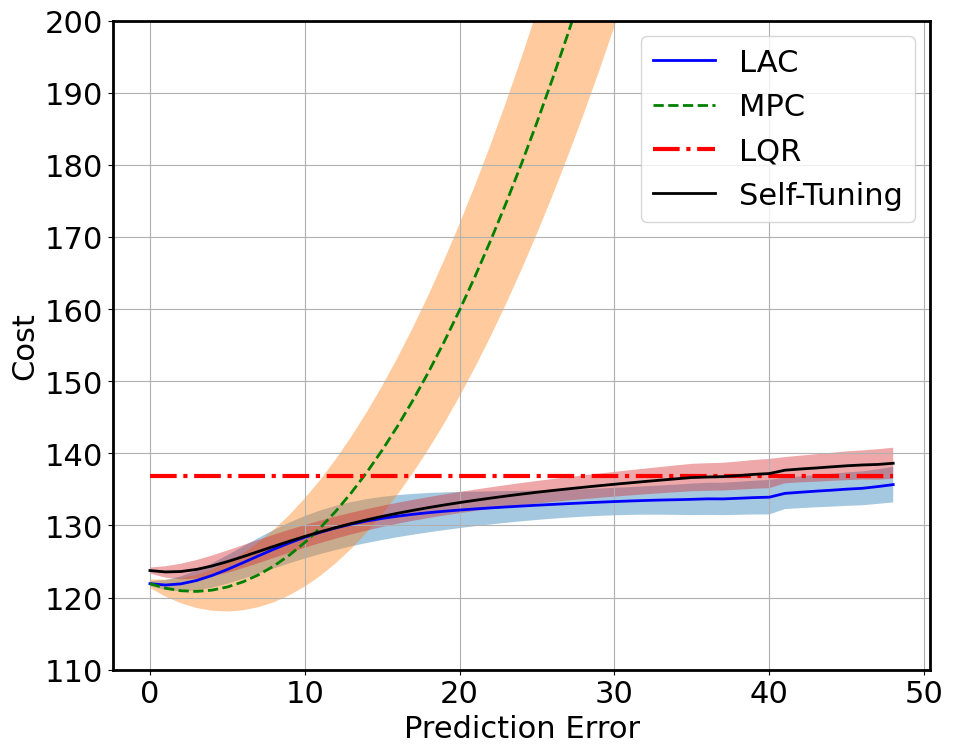}
    \caption{Illustrative comparison of closed-loop control cost versus prediction error for Learning-Augmented Control (\ouralg, this work), \texttt{MPC}, \texttt{LQR}, and \texttt{Self-Tuning}~\cite{li2022robustness}. \ouralg (blue solid) and \texttt{Self-Tuning} (black solid) exemplify the desired ``best-of-both-worlds'' behavior: they achieve near-optimal low cost comparable to \texttt{MPC} when predictions are good (low error) and robustly maintain cost close to the \texttt{LQR} baseline when prediction quality is poor (high error). Shaded regions represent standard deviation over multiple simulation runs. In Section~\ref{sec:linear_settings}, detailed experimental settings are provided.}
    \vspace{-10pt}
    \label{fig:lqc}
\end{figure}
% Take a time-invariant linear system system~\cite{kalman1964linear} whose discrete-time version has been studied recently through the lens of competitive analysis\cite{yu2022competitive,goel2022competitive,li2022robustness} as a concrete example. The offline optimization is represented as
% \begin{subequations}
% \begin{align}
% \label{eq:lqc_1}
% \min_{u_{0:\nt-1}} \sum_{t=0}^{\nt-1}& \left(x_t^\top Q x_t + u_t^\top R u_t\right) + x_{\nt}^\top P x_{\nt}\\
% \label{eq:lqc_2}
% & x_{t+1} = Ax_t+Bu_t+\phi_t^{\star}
% \end{align}
% \end{subequations}
% with known stabilizable matrices $A\in  \mathbb{R}^{n\times n}$, $B\in \mathbb{R}^{n\times m}$, and unknown perturbations $\left(\phi_t^{\star}:t\in [\nt]\right)$. Here, $Q\in\mathbb{R}^{n\times n}, R\in\mathbb{R}^{m\times m}$ are positive-definite matrices, and $P$ is a symmetric positive-definite cost-to-go solution of the following discrete algebraic Riccati equation (DARE).\footnote{There exists a unique solution of the DARE $P=Q+A^\top P A - A^\top PB (R+B^\top P B)^{-1} B^\top PA$ is guaranteed because $(A,B)$ is stabilizable and $Q,R$ are positive-definite~\cite{dullerud2013course}.} 
An example is shown in Figure~\ref{fig:lqc}, where a
Linear Quadratic Regulator (\texttt{LQR}, red dash-dot) uses only baseline nominal predictions by regarding future system disturbances as zero, yielding \textit{robust} but potentially suboptimal cost. Classic \texttt{MPC} (green dashed) directly uses potentially inaccurate predictions, achieving \textit{consistent} performance at low error but suffering significant performance degradation (high cost and variance) as error increases. Classic \texttt{MPC} and \texttt{LQR} cannot achieve (1) and (2) simultaneously.

% A core aspect of MPC is its reliance on predictions of system behavior to optimize control actions. However, in the AI era, these predictions often originate from inherently uncertain sources from various off-the-shelf machine learning models. For instance, auxiliary advice is employed as future information in mobile robot control~\cite{chipalkatty2013less} and battery management~\cite{li2024out}, while machine learning (ML)-generated predictions are increasingly utilized in data-driven MPC~\cite{rosolia2018data,hewing2020learning,bold2024data} and model-based reinforcement learning~\cite{park2023predictable}. 
% Despite their widespread adoption, these prediction sources frequently lack guarantees on accuracy, a limitation that becomes particularly critical in dynamic, time-varying environments. 

% \textbf{Related Work.}
Focusing on this need, various new approaches have sought to strike a balance between uncertain ML predictions and safer nominal ones. A preliminary algorithm termed \texttt{Self-Tuning} (black solid in Figure~\ref{fig:lqc}) was presented in~\cite{li2022robustness}, but its applicability is limited to time-invariant linear dynamics with quadratic costs, and its competitive ratio guarantee depends on quantifying the variation of predictions and system perturbations due to its Follow-the-Leader (FTL) basis. For more general time-varying systems, online policy selection via online optimization has been explored in~\cite{lin2023online} by tuning policy parameters. An adaptive regret of $O(\sqrt{\nt})$ is guaranteed by~\cite{lin2023online}, but the surrogate costs need to be convex and differentiable with respect to policy parameters. Recent results such as~\cite{shen2024combining} have used the disturbance response control to ensure the convexity of the stable controller space;~\cite{li2024disentangling} has applied online optimization algorithms to linear quadratic control problems with latent disturbances, yet the analysis requires additional assumptions to stabilize the key confidence parameter.
Consequently, there remains a need for a principled method with theoretical guarantees that achieves the dual ``best-of-both-worlds''  objectives above for general dynamical systems subject to uncertain ML predictions.

Our work also contributes to the evolving field of adaptive control, especially in its recent intersection with learning theory to address non-asymptotic metrics. While the adaptive control community has traditionally concentrated on Lyapunov stability and asymptotic convergence~\cite{slotine1991applied,dhar2021indirect,fukushima2007adaptive}, the newer trend represented in~\cite{abbasi2011regret},\cite{simchowitz2020naive},\cite{dean2019sample}, and~\cite{agarwal2019online} employs learning-theoretic measures such as competitive ratio and regret for finite horizon analysis. In particular,~\cite{abbasi2011regret,agarwal2019online,lu2025almost,zhao2025data} have applied regret minimization to provide finite-time performance bounds and~\cite{dean2019sample,simchowitz2020naive,karapetyan2025closed,muthirayan2025meta} have established sample complexity guarantees for online control. In contrast to those recent results that rely on statistical assumptions and weaker regret metrics, our work employs competitive analysis to design controllers robust to adversarial disturbances and evaluated against a true offline optimal.

\textbf{Contributions.}
To tackle the challenge of utilizing untrusted ML predictions in MPC, this work studies the \textit{Learning-Augmented Control } (LAC) problem for general time-varying nonlinear systems with constraints, where a controller receives potentially inaccurate, parameterized predictions of future system parameters $(\phi^\star_{\tau}:\tau\geq t)$ within a receding horizon framework. The core problem is that the prediction error  is unknown until the corresponding true parameters are revealed. Our framework aims to transform conventional MPC by introducing a robust layer of online confidence learning. 
In the first part
of this article, we propose a $\lambda_t$-confident control law, where $\lambda_t$ is a time-varying confidence parameter learned online. This mechanism explicitly balances leveraging potentially high-reward but untrusted ML predictions against relying on safe but conservative nominal predictions. We formulate the online learning of the confidence sequence $\{\lambda_t\}_{t=0}^{T-1}$ as a \textit{delayed online convex optimization} problem. We derive a per-step surrogate error bound, minimizing which aims to minimize the overall control cost. Crucially, evaluating the bound and its gradient at time $t$ requires knowledge of prediction errors up to time $t+k$ where $k$ is a prediction window size, information which only becomes available at time $t' = \min\{t+k, T\}$. This inherent \textit{delay of $k$ steps} necessitates a specialized online learning approach for the predictive control paradigm. We develop a \dcl algorithm (Algorithm~\ref{alg:dcl}), based on online mirror descent (or FTRL), which updates $\lambda_t$ using feedback received up to time $t-k$. This involves analyzing the problem as $k$ interleaved online learning subproblems.

Focusing on finite horizon competitive  analysis,  we provide theoretical guarantees for the \ouralg policy under standard MPC assumptions (SSOSC, LICQ, leading to EDPB in~\Cref{def:perturbation_bound}). We establish a competitive ratio bound, which formally characterizes the desired behavior: performance is near-optimal when cumulative prediction errors are small, and degrades gracefully relative to the errors otherwise (see \Cref{thm:lac}). Unlike existing self-tuning or robust control methods (e.g., \cite{li2022robustness},\cite{ mayne2005robust}, \cite{calafiore2012robust}, and~\cite{ berberich2020data}), the \ouralg framework  adapts directly to the realized errors. Our analysis for the LQC setting improves upon the bounds in \cite{li2022robustness} by removing dependency on perturbation variation terms, as shown in~\Cref{thm:lac_lqc}, and~\ref{thm:lqc_lower_bound}.

\textbf{Notation.}
We use $[t_1:t_2]$ to denote the set of integers $\{t_1, \ldots, t_2\}$. The symbol $\|\cdot\|$ denotes the $\ell_2$-norm and for a square matrix $M$, $\|\cdot\|_{M}^2\coloneqq (\cdot)^{\top}M(\cdot)$. Given a prediction window size $k$ and time horizon $T$, we define the time indices $\underline{t} \coloneqq \max\{t-k+1, 0\}$, $\overline{t} \coloneqq \min\{t+k-1, T-1\}$, and $t' \coloneqq \min\{t+k, T\}$. The set $\mathcal{I}$ represents the unit interval $[0,1]$, and $\mathbb{N}$ denotes the set of non-negative integers. For a set $\mathsf{S}$, $\mathsf{S}+a\coloneqq \{s\in\mathsf{S}:s+a\}$ with $a\in\mathbb{R}$. The largest and smallest eigenvalues are written as $\overline{\lambda}(\cdot)$ and $\underline{\lambda}(\cdot)$. The mappings $\overline{\sigma}(\cdot)$ and $\underline{\sigma}(\cdot)$ represent the largest and smallest singular values of a matrix, respectively.  To simplify the presentation, we adopt the asymptotic notation (Big O) $O,\Omega,\Theta$, $\omega$, and $o$ with respect to the problem horizon $\nt$.

\begin{table*}[ht]
\centering
\renewcommand{\arraystretch}{1.5} % Increased spacing for readability
\begin{tabular}{l|l|l|l|l}
\specialrule{.09em}{.1em}{.1em}
 \textbf{System} & \textbf{Predictions} & \textbf{Nominal Values} & \multicolumn{2}{c}{\textbf{Competitive Ratio (CR) Bounds}} \\
\hline
\multirow{2}{*}{LQC} & $(\phi_t:t\in [\nt])$ given at $t=0$ & $(\kappa_t=0:t\in [\nt])$ & Theorem 4.1~\cite{li2022robustness} & $\mathsf{CR}(\pi) \leq 1+\frac{O(\|\boldsymbol{\varepsilon}\|)}{1+\Omega(\|\boldsymbol{\varepsilon}\|)}+O\left(\mu_{\mathrm{Var}}\right)$ \\
\cline{2-5} % Separates the two LQC cases
 & $(\phi_{t:\overline{t}|t}:t\in [\nt])$ (receding horizon) & $(\kappa_{t:\overline{t}|t}=0:t\in [\nt])$ & \textbf{Theorem 2 (This Work)} & $\mathsf{CR}(\pi_{\ouralg}) \leq 1+O\left(\frac{\|\boldsymbol{\varepsilon}\|}{\|\boldsymbol{\varepsilon}\|+\sqrt{\nt\log\nt}}\right) + o(1)$ \\
\hline
\multirow{2}{*}{General} & \multirow{2}{*}{$(\phi_{\tau|t}: t\leq \tau\leq \overline{t})$ (receding horizon)} & \multirow{2}{*}{$(\kappa_{\tau|t}: t\leq \tau\leq \overline{t})$} & \multirow{2}{*}{\textbf{Theorem 1 (This Work)}} & $\mathsf{CR}(\pi_{\ouralg})\leq 1+2\sqrt{\zeta}+\zeta$ \\
& & & & (Asymptotic) $\leq 1 +O\left(\sqrt{\varpi}+\varpi\right)$ \\
\specialrule{.09em}{.1em}{.1em}
\end{tabular}
\caption{Comparison of competitive ratio bounds for learning-augmented control. Our results (Theorems 1 and 2) provide guarantees for general nonlinear systems with receding horizon predictions and improve upon existing bounds for LQC systems that use one-shot predictions by removing dependency on variation terms ($\mu_{\mathrm{Var}}$). The terms $\zeta$ and $\varpi$ are defined in Section~\ref{sec:main}, Equation~\eqref{eq:def_zeta}.}
\label{tab:results_comparison}
\end{table*}

\section{Preliminaries and Problem Formulation}
\label{sec:problem_formulation}
We consider a finite-time discrete dynamical system, with a set $[\nt]\coloneqq \{0,\ldots,\nt-1\}$ denoting all time steps. Let $x_t\in\mathbb{R}^{\nn}$ be a system state and $u_t\in\mathbb{R}^{\nm}$ an action from a state feedback control policy $\pi_t$ at time $t\in [\nt]$. A controller $\pi\coloneqq (\pi_t:t\in [\nt])$ is a collection of control policies for all $t$. The system update rule is given by
\begin{align}
\label{eq:dynamics}
   x_{t+1} = f_t\left(x_t,u_t;\phi_t^{\star}\right), \quad t\in [\nt], 
\end{align}
subject to
\begin{align}
\label{eq:constraints}
h_t\left(x_t,u_t;\phi_t^{\star}\right)\leq 0, \
x_t\in\mathsf{X}, \ u_t\in\mathsf{U}, \quad t\in [\nt],
\end{align}
where $\mathsf{X}\subseteq\mathbb{R}^{\nn}$ and $\mathsf{U}\subseteq\mathbb{R}^{\nm}$. An initial state $x_0\in\mathsf{X}$ is given at time $t=0$.
The total cost of a controller $\pi\coloneqq (\pi_t:t\in [\nt])$ is
\begin{align}
\label{eq:costs}
       J(\pi) \coloneqq \sum_{t=0}^{\nt-1}c_t\left(x_t,u_t;\phi_t^{\star}\right) + c_{\nt}\left(x_{\nt};\phi_T^{\star}\right).
\end{align}

The dynamics of the system $f_t:\mathbb{R}^{\nn}\times\mathbb{R}^{\nm}\times \Phi\rightarrow \mathbb{R}^\nn$ in~\eqref{eq:dynamics}, cost $c_t:\mathbb{R}^{\nn}\times\mathbb{R}^{\nm}\times \Phi\rightarrow \mathbb{R}^\nn$, and the terminal cost $c_{\nt}:\mathbb{R}^{\nn}\times \Phi\rightarrow \mathbb{R}^\nn$ in~\eqref{eq:costs} are parameterized by unknown variables in a convex \textit{uncertainty set} $0\in\Phi\subseteq\mathbb{R}^{d}$ with a diameter $\gamma>0$, e.g., $\sup_{\phi,\varphi\in\Phi}\|\phi-\varphi\|\leq\gamma$. 
At each time $t\in [\nt]$, after receiving $x_t$ and before generating an action $u_t$, predictions of future ground-truth parameters $\phi^{\star}_{t:\overline{t}}\coloneqq \left(\phi_t^{\star},\ldots,\phi_{\overline{t}}^{\star}\right)$ are given and written as $\phi_{t:\overline{t}|t}\coloneqq (\phi_{t|t},\ldots,\phi_{\overline{t}|t})$ (denoting $\smash{\overline{t}\coloneqq\min\{t+\nk-1,\nt-1\}}$). In addition, the parameters have nominal values $\kappa_{t:\overline{t}|t}\coloneqq (\kappa_{t|t},\ldots,\kappa_{\overline{t}|t})$.
The true values of $\left(\phi^{\star}_{\tau}:\tau\leq t\right)$ are revealed to the controller at time $t+1$.
% We consider two information models in this work: (1). The true values of $\left(\phi^{\star}_{\tau}:\tau\leq t\right)$ are revealed to the controller at time $t+1$;
% (2). Only estimates of ground-truth parameters $\left(\phi^{\star}_{\tau}:\tau\leq t\right)$ are available at time $t+1$.

Let $\mathbf{f}\coloneqq (f_t:t\in [\nt])$, $\mathbf{h}\coloneqq (h_t:t\in [\nt])$, and $\mathbf{c}\coloneqq (c_t:t\in [\nt]\cup\{\nt\})$.
Throughout this paper, we assume the following standard requirements on the system. 

\begin{assumption}[Stabilizability and $C^2$]
\label{ass:basic}
The system $(\mathbf{f},\mathbf{h})$ is stabilizable.
All cost functions in $\mathbf{c}$, dynamical functions in $\mathbf{f}$, and constraints in $\mathbf{h}$ are twice continuously differentiable ($C^2$) with respect to actions, states, and unknown variables. The cost functions are nonnegative, convex, and $\ell$-smooth.
\end{assumption}

% Since $f_t$ and $c_t$ are $C^2$ with respect to $u_{t}\in\mathsf{U}$ and $(x_t,u_t)\in\mathsf{X}\times\mathsf{U}$, it is Lipschitz continuous and there exist constants $L_f,L_c>0$ such that for any $u,u'\in\mathsf{U}$, $x,x'\in\mathsf{X}$, and $t\in [\nt]$,
% \begin{align*}
%    \left| f_t\left(x,u;\phi_t^{\star}\right) -f_t\left(x,u';\phi_t^{\star}\right)\right| &\leq L_f\left\|u-u'\right\| \ \text{for all } x\in\mathsf{X},\\
%    \left| c_t\left(x,u;\phi_t^{\star}\right) -c_t\left(x',u';\phi_t^{\star}\right)\right| & \leq L_c\left(\left\|x-x'\right\|+\left\|u-u'\right\|\right).
% \end{align*}

We further assume the existence of a convex control-invariant set with respect to the uncertainty set $\Phi$, which can be constructed by standard robust control techniques.
% , as we elaborate in a concrete robot arm control example in Section~\ref{sec:robot_arm_model}.

\begin{assumption}
[Control-Invariance]\label{ass:control_invariant}
The set ${\mathsf{X}}$ is control invariant such that for any $x\in {\mathsf{X}}$, we can find some $u\in\mathsf{U}$ satisfying $h_t(x,u;\phi)\leq 0$ and $f_t(x,u;\phi)\in{\mathsf{X}}$ for any $t\in [\nt]$ and $\phi\in\Phi$.  Moreover, the convex uncertainty set $\Phi$ is robust such that for any $t\in [\nt]$ if $h_t(x,u;\varphi)\leq 0$ and $f_t(x,u;\varphi)\in{\mathsf{X}}$ for some $x\in\mathsf{X}$, $u\in\mathsf{U}$, and $\varphi$, then $h_t(x,u;\phi)\leq 0$ and $f_t(x,u;\phi)\in{\mathsf{X}}$ for all $\phi\in\Phi$. 
\end{assumption}

\subsection{Preliminaries: MPC, Assumptions, and LQC}

\subsubsection{Model Predictive Control} 
\label{sec:model_mpc}
Given predictions in $\phi_{t:\overline{t}|t}$ at time $t\in [\nt]$, a natural scheme to obtain an action $u_t$ is the MPC scheme~\cite{garcia1989model,zheng1995stability,lin2022bounded}, described by the following optimization:
\begin{subequations}
\begin{align}
\label{eq:mpc1}
\min_{u_{t:\overline{t}},x_{t+1:t'}}  \sum_{\tau=t}^{\overline{t}}  c_{\tau}(x_\tau,u_\tau;\phi_{\tau|t}) + & c_{\nt}(x_{t'};\phi_{\nt|t})\\
\label{eq:mpc2}
    \text{s.t. }  x_{\tau + 1} =  f_{\tau}(x_{\tau},u_{\tau};\phi_{\tau|t}), \ & \tau=t,\ldots,\overline{t},\\
\label{eq:mpc3}
0 \geq  h_{\tau}(x_{\tau},u_{\tau};\phi_{\tau|t}), \ & \tau=t,\ldots,\overline{t},\\
\label{eq:mpc4}
x_{t'},  x_\tau\in\mathsf{X}, \ u_{\tau}\in\mathsf{U}, \ & \tau=t,\ldots,\overline{t}.
\end{align}
\label{eq:mpc}
\end{subequations}
In~\eqref{eq:mpc1}-\eqref{eq:mpc4}, a prediction window size $\nk$ is used and the terminal cost $c_{\nt}(x_{t'})$ depends on the state $x_{t'}$. 
% with a predicted parameter $\phi_{t'}$
For notational convenience, we denote the optimal solution corresponding to $u_\tau$ with $t\leq\tau\leq\overline{t}$ of the optimization above by $u_\tau(x_t;\phi_{t:\overline{t}|t})$ where $\phi_{t:\overline{t}|t}$ encapsulates predictions of future $\nk$ dynamics and costs, as well as $\phi_{\nt|t}$. It is further simplified as $u_t(\phi_{t:\overline{t}|t})$ if there is no ambiguity on the initial state $x_t$. Similarly, we let $x_\tau(x_{t};\phi_{t:\overline{t}|t})$ with $t+1\leq\tau\leq t'$ be the corresponding state as an optimal solution to~\eqref{eq:mpc}.

%[Bounded Perturbation Property]
% \begin{definition}
% \label{def:deb}
% An MPC control policy~\eqref{eq:mpc1}-\eqref{eq:mpc2} corresponding to the control problem instance $(f,c)$ satisfies the \emph{Bounded Perturbation Property (BPP)}  if
% there exists a mapping $\rho:\mathbb{N}\rightarrow\mathbb{R}_+$ with $\sum_{t\in [\nt]}\rho(t)\leq C$ for some constant $C>0$ such that
% $\eta_t \leq \sum_{\tau=t}^{\overline{t}} \rho(\tau-t) \varepsilon_{t|\tau} + \rho(\nk)
% $
% holds for all $t\in [\nt]$.
% \end{definition}

% \paragraph{Nominal MPC} When using the nominal values $\kappa_{t:\overline{t}|t}$ instead of $\phi_{t:\overline{t}|t}$, the optimal solution of~\eqref{eq:mpc1}-\eqref{eq:mpc4} becomes $u_t(\kappa_{t:\overline{t}|t})$. We call such an instance a nominal MPC (denoted as \textsc{nMPC}).

% \paragraph{Classic MPC} Similarly, with $\phi_{t:\overline{t}t}$ being used, we term such a scheme a classic MPC (denoted as \textsc{cMPC}).

\subsubsection{Second and First-Order Assumptions}
\label{sec:second_first_assumption}
We simplify the primal variables and dual variables in~\eqref{eq:mpc} as $\mathbf{p}_t\coloneqq (u_{t:\overline{t}},x_{t+1:t'})$ and $\mathbf{d}_t$, respectively. The input parameter is $\mathbf{r}_t\coloneqq (\phi_{t:\overline{t}},x_t)$.
Consider the Lagrangian $\mathcal{L}\left(\mathbf{p}_t,\mathbf{d}_t;\mathbf{r}_t\right)$ of the optimization~\eqref{eq:mpc1}-\eqref{eq:mpc3}. 
We present regular assumptions on the MPC formulation~\eqref{eq:mpc1}-\eqref{eq:mpc4}. We define $\mathbf{z}_t(\mathbf{p}_t)$ as a \textit{primal-dual} solution of~\eqref{eq:mpc} if $\mathbf{p}_t$ satisfies the first-order optimality conditions with Lagrange multiplier $\mathbf{d}_t$. Abusing the notation, the primal-dual variable is written as $\mathbf{z}_t\coloneqq (\mathbf{p}_t,\mathbf{d}_t)$.

\begin{assumption}[MPC Regularity~\cite{robinson1980strongly,shin2022exponential}]
    \label{ass:mpc}
For all $t\in [\nt]$, $\phi\in\Phi$, and $x_t\in\mathsf{X}$, the primal-dual solution $\mathbf{z}_t$ of the MPC scheme~\eqref{eq:mpc1} satisfies SSOSC
and LICQ.
\end{assumption}

The Strong Second Order Sufficient Conditions (SSOSC) ensure that the  
reduced Hessian is positive-definite~\cite{robinson1980strongly,shin2022exponential}, i.e., the Hessian $\nabla_{\mathbf{p}_t\mathbf{p}_t}^2\mathcal{L}\left(\mathbf{z}_t;\mathbf{r}_t\right)$ is positive-definite on the null space defined by~\eqref{eq:mpc2} and active constraints in~\eqref{eq:mpc3} with nonzero duals. Moreover, Linear Independence Constraint Qualification (LICQ)~(see discussions in~\cite{robinson1980strongly,shin2022exponential})  ensures the multipliers are unique and the KKT system is well-posed. The LICQ requires that the gradients in the constraint Jacobian defined by equality and active inequality constraints are linearly independent at the optimal primal-dual solution. In particular, borrowing the notation in~\cite{shin2022exponential}, these conditions are given by
\begin{align}
\label{eq:ssosc}
\tag{SSOSC}
\mathrm{ReH}\left(\nabla_{\mathbf{p}_t\mathbf{p}_t}^2\mathcal{L}\left(\mathbf{z}_t(\mathbf{r}_t);\mathbf{r}_t\right),\nabla_{\mathbf{p}_t}\overline{\mathbf{g}}\left(\mathbf{p}_t;\mathbf{r}_t\right)\right) &\succ 0,\\
\label{eq:licq}
\tag{LICQ}
\underline{\sigma}\left(\nabla_{\mathbf{p}_t}\mathbf{g}\left(\mathbf{p}_t;\mathbf{r}_t\right)\right) &> 0,
\end{align}
where $\nabla_{\mathbf{p}_t}\mathbf{g}\left(\mathbf{p}_t;\mathbf{r}_t\right)$  collects the gradients of active constraints $\nabla_{\mathbf{p}_t}f_{\tau}(\mathbf{p}_t;\phi_{\tau|t},x_t)$ and $\nabla_{\mathbf{p}_t}h_{\tau}(\mathbf{p}_t;\phi_{\tau|t},x_t)$ in~\eqref{eq:mpc} where $h_{\tau}(x_{\tau},u_\tau;\phi_{\tau|t})=0$, for $\tau=t,\ldots,\overline{t}$; $\nabla_{\mathbf{p}_t}\overline{\mathbf{g}}\left(\mathbf{p}_t;\mathbf{r}_t\right)$ contains a subset of rows of $\nabla_{\mathbf{p}_t}\mathbf{g}\left(\mathbf{p}_t;\mathbf{r}_t\right)$ by enforcing the corresponding entries in the Lagrangian multiplier $\mathbf{d}_t$ to be nonzeros. The reduced Hessian operator is defined as $\mathrm{ReH}(H,M)\coloneqq Z_M^{\top} H Z_M$ where $Z_M$ is the null-space matrix of $M$.

% Let $z_t\coloneqq u_t$ and $z_\tau\coloneqq (x_\tau,u_\tau)$ for all $\tau\leq $, $t\in [\nt]$ be the concatenation of variables in~\eqref{eq:mpc}.

Note that for Linear Time-Varying (LTV) systems, uniform controllability implies the LIQC~\cite{shin2021controllability}. We further assume the following singular spectrum bounds on the Hessian sub-matrices of the Lagrangian hold, which are necessary to guarantee
sufficiently fast decay~\cite{shin2022exponential}. Let $H_1(\mathbf{r}_t)\coloneqq \nabla_{\mathbf{z}_t\mathbf{z}_t}^2\mathcal{L}\left(\mathbf{z}_t(\mathbf{r}_t);\mathbf{r}_t\right)$ and $H_2(\mathbf{r}_t)\coloneqq\nabla_{\mathbf{z}_t\mathbf{p}_t}^2\mathcal{L}\left(\mathbf{z}_t(\mathbf{r}_t);\mathbf{r}_t\right)$.

\begin{assumption}[Singular Spectrum Bounds~\cite{shin2022exponential}]
\label{ass:singular_bounds}
For all $t\in [\nt]$, $\phi\in\Phi$, and $x_t\in\mathsf{X}$, the Hessian matrices $H_1(\mathbf{r}_t)$ and $H_2(\mathbf{r}_t)$ satisfy
    \begin{align}
        \label{eq:hessian_1}
    \overline{\sigma}\left(H_1(\mathbf{r}_t)\right), \overline{\sigma}\left(H_2(\mathbf{r}_t)\right) &\leq \overline{\overline{\sigma}},\\
        \label{eq:hessian_2}
     \underline{\sigma}\left(H_1(\mathbf{r}_t)[\mathcal{B},\mathcal{B}]\right) &\geq \underline{\underline{\sigma}}, \text{ for all } \mathcal{B}_{>0}\subseteq\mathcal{B}\subseteq\mathcal{B}_{\geq 0}
    \end{align}
    for some constants $\overline{\overline{\sigma}},\underline{\underline{\sigma}}>0$, where $H_1(\mathbf{r}_t)[\mathcal{B},\mathcal{B}]\coloneqq (H_1(\mathbf{r}_t)(ij):i,j\in\mathcal{B})$ is a submatrix of $H_1(\mathbf{r}_t)$, with $\mathcal{B}$ defined as a set of indices that consists of $\mathcal{B}_{>0}$, a set of primal-dual indices corresponding to all primal entries and nonzero dual entries of active constraints. Similarly, $\mathcal{B}_{\geq 0}$ is a set of primal-dual indices corresponding to all primal entries and dual entries of active constraints.
\end{assumption}

Note that~\eqref{eq:hessian_1} and~\eqref{eq:hessian_2} together imply the bounds defined by (4.16a), (4.16b) and (4.16c) in Theorem 4.5 in~\cite{shin2021controllability}, since the spectral norm of a matrix is always larger than that of a sub-matrix, and permutation does not change the singular value of a matrix. 
The following property summarizes a key bound on the per-step action error (see~\Cref{eq:per_step_error} in Appendix~\ref{app:proof_lemma_regret_error}).

% Together, Assumption~\ref{ass:basic} and~\ref{ass:basic} imply a \textit{sensitivity bound}, i.e., there exists a constant $C_{u}>0$ such that for any $\phi_{\tau},\varphi_\tau\in\Phi$ (with $\tau=t,\ldots,t'$)
% \begin{align}
% \label{eq:sensitivity_bound_action}
% \left\|u_t\left(\phi_{t:\overline{t}}\right) - u_t\left(\varphi_{t:\overline{t}}\right)\right\|\leq C_{u}\left\|\phi_{t:\overline{t}}-\varphi_{t:\overline{t}}\right\|,
% \end{align}
% as a result of the implicit function theorem.

\begin{definition}[EDPB~\cite{shin2022exponential,lin2022bounded}]
\label{def:perturbation_bound}
An MPC control scheme~\eqref{eq:mpc1}-\eqref{eq:mpc4} corresponding to the control problem instance $(\mathbf{f},\mathbf{c})$ satisfies the \emph{Exponentially Decaying Perturbation Bounds (EDPBs)} if there exists a mapping $\rho:\mathbb{N}\rightarrow\mathbb{R}_+$ with $\sum_{t\in [\nt]}\rho(t)\leq C$ for some constant $C>0$ such that for any $x,x'\in\mathsf{X}$, $t\in [\nt]$, $\tau\in [t,\overline{t}]$, $\tau'\in [t+1,t']$, and $\phi_{t:\overline{t}}\in\Phi^{\overline{t}-t+1}$,
\begin{align}
\label{eq:perturbation_bound_action}
\left\|u_{\tau}\left(x;\phi_{t:\overline{t}}\right)-u_{\tau}\left(x';\phi_{t:\overline{t}}\right)\right\| &\leq \rho(\tau-t)\left\|x-x'\right\|,\\
\label{eq:perturbation_bound_state}
\left\|x_{\tau'}\left(x;\phi_{t:\overline{t}}\right)-x_{\tau'}\left(x';\phi_{t:\overline{t}}\right)\right\| &\leq \rho(\tau'-t)\left\|x-x'\right\|,
\end{align}
and similarly for any $\phi_{t:\overline{t}},\varphi_{t:\overline{t}}\in\Phi^{\overline{t}-t+1}$ and $t\in [\nt]$,
\begin{align}
\label{eq:sensitivity_bound_action}
\left\|u_t\left(\phi_{t:\overline{t}}\right) - u_t\left(\varphi_{t:\overline{t}}\right)\right\| \leq & \sum_{\tau=t}^{\overline{t}} \rho(\tau-t) \|\phi_\tau-\varphi_\tau\|,\\
\label{eq:sensitivity_bound_state}
\left\|x_{t+1}\left(\phi_{t:\overline{t}}\right) - x_{t+1}\left(\varphi_{t:\overline{t}}\right)\right\| \leq & \sum_{\tau=t}^{\overline{t}} \rho(\tau-t) \|\phi_\tau-\varphi_\tau\|.
\end{align}
\end{definition}

With the singular spectrum bounds in Assumption~\ref{ass:singular_bounds}, SSOSC and LICQ (Assumption~\ref{ass:mpc}) imply that the optimal MPC solutions satisfy Definition~\ref{def:perturbation_bound} (see~\cite{shin2022exponential,lin2021perturbation,lin2022bounded}). 
\begin{lemma}[\cite{shin2022exponential,lin2021perturbation,lin2022bounded}]
    \label{lemma:edpb}
    Under Assumption~\ref{ass:basic},~\ref{ass:mpc} and~\ref{ass:singular_bounds}, the MPC scheme~\eqref{eq:mpc} satisfies the EDPBs in Definition~\ref{def:perturbation_bound}.
\end{lemma}
This lemma follows directly from Theorem 4.5 in~\cite{shin2021controllability} (see also Theorem H.1 of~\cite{lin2022bounded}); therefore, its proof is omitted.
Furthermore, the required assumptions in Lemma~\ref{lemma:edpb} can be relaxed to only hold locally for a base input in $\Phi$, using specifically designed terminal costs in~\eqref{eq:mpc1} (see~\cite{lin2022bounded}). To simplify the presentation, we assume Definition~\ref{def:perturbation_bound} holds uniformly. Note that the EDPBs in Definition~\ref{def:perturbation_bound} naturally hold for various special dynamical systems such as Linear Quadratic Control (LQC), and LTV. In the sequel, we use LQC to exemplify the general notions of MPC and EDPB.

%sensitivity bound in~\eqref{eq:sensitivity_bound_action} and 

% Note that Assumption~\ref{ass:basic} and Assumption~\ref{ass:basic} are weaker than the sufficient assumptions (Assumption H.1.) for EDPB in~\cite{lin2022} with respect to both states and predictive variables. In this work, we only require the EDPB with respect to states to hold.

% Remark: LQC as an example; Recursive feasibility. 

\subsubsection{Linear Quadratic Control (LQC)}
\label{sec:examples}

% We consider a linear dynamic system with adversarial perturbations,
% \begin{align}
% \label{eq:linear_dynamic}
%     x_{t+1} = Ax_t + B u_t + w_t, \text{ for } t = 0,\ldots,\nt-1,
% \end{align}
% where $A\in\mathbb{R}^{\nn\times\nn}$ and $B\in\mathbb{R}^{\nn\times\nm}$, and $w_t\in\mathbb{R}^{\nn}$ denotes some unknown perturbation chosen adversarially. 

The general dynamics in~\eqref{eq:dynamics} can be instantiated to the standard linear control system~\cite{kalman1964linear} whose discrete-time version has been studied recently through the lens of competitive analysis\cite{yu2022competitive,goel2022competitive,li2022robustness} by setting 
\begin{align}
\label{eq:dynamics_lqc}
f\left(x_t,u_t;\phi_t^{\star}\right)\coloneqq Ax_t+Bu_t+\phi_t^{\star}
\end{align}
with known system matrices $A\in  \mathbb{R}^{n\times n}$, $B\in \mathbb{R}^{n\times m}$, and unknown perturbations $\left(\phi_t^{\star}:t\in [\nt]\right)$. 
% We consider the standard regime that the pair $(A,B)$ is stabilizable~\cite{shin2023near,li2022robustness}.
% Without loss of generality, we also assume the system is initialized with some fixed $x_0\in\mathbb{R}^{n}$. 
The control objective is the following quadratic costs given matrices $A,B,Q,R$:
\begin{align}
\label{eq:cost_lqc}
    J(\pi) =  \sum_{t=0}^{\nt-1} \left(x_t^\top Q x_t + u_t^\top R u_t\right) + x_{\nt}^\top P x_{\nt},
\end{align}
where $Q\in\mathbb{R}^{n\times n}, R\in\mathbb{R}^{m\times m}$ are positive-definite matrices, and $P$ is a symmetric positive-definite cost-to-go solution of the following Discrete Algebraic Riccati Equation (DARE), which must exist because $(A,B)$ is stabilizable and $Q,R$ are positive-definite  \cite{dullerud2013course}:
\begin{align*}
    P=Q+A^\top P A - A^\top PB (R+B^\top P B)^{-1} B^\top PA.
\end{align*}
Given $P$,  we define $K \coloneqq (R+B^\top P B)^{-1} B^\top P A$, and $u_t=-Kx_t$ is the feedback control law of a Linear Quadratic Regulator (LQR), which is optimal in the case of zero disturbances ($\phi_t^{\star} =0$ for all $t\in [\nt]$). 
Furthermore, let $F\coloneqq A-BK$ be the closed-loop system matrix. The Gelfand’s formula implies that there must exist a constant $C_F>0$ and a spectral radius $\rho_F\in(0,1)$ such that $\Vert F^t\Vert\leq C_F\rho_F^t$ for all $t\in [\nt]$. 
For LQC, setting a terminal cost $c(x_{t'})=x_{t'}^{\top}Px_{t'}$ in~\eqref{eq:mpc1}, $u_t(x_t;\phi_{t:\overline{t}|t})$ can be explicitly written as (see~\cite{yu2020power,li2022robustness})
\begin{align}
\nonumber
u_t(x_t;\phi_{t:\overline{t}|t})=-Kx_t -(R+B^{\top}PB)^{-1}B^{\top}\sum_{\tau=t}^{\overline{t}}(F^{\top})^{\tau-t}P\phi_{\tau|t}.
\end{align}
For the LQC case, this leads to the following explicit EDPB, which is a specific instance of the sensitivity bound in \eqref{eq:sensitivity_bound_action}:
\begin{align*}
\left\|u_t\left(\phi_{t:\overline{t}}\right) - u_t\left(\varphi_{t:\overline{t}}\right)\right\|\leq & C_F\left\|(R+B^{\top}PB)^{-1}B^{\top}\right\|\|P\|\\ &\cdot\sum_{\tau=t}^{\overline{t}}\rho_F^{\tau-t}
    \left\|\phi_{\tau}-\varphi_{\tau}\right\|.
\end{align*}
Bounds for the state trajectories and with respect to initial state perturbations can be derived in a similar manner.
\subsection{Problem Formulation}
\label{sec:problem_formulation_error}
% We assume that there exists a constant $\gamma>0$ such that $\|\phi_t^{\star}\|,\|\phi_{\tau|t}\|,\|\kappa_{\tau|t}\|\leq \gamma$ for all $t\in [\nt]$ and $\tau\in [t:\overline{t}]$.
To evaluate the performance, for each time $t\in [\nt]$, given $\phi_{t:\overline{t}}$ and $\phi_{t:\overline{t}}^{\star}$, we define the following step-wise \textit{prediction errors} corresponding to the \textit{predictive} parameters and \textit{nominal} parameters respectively:
\begin{align}
    \label{eq:prediction_error}
    \epsilon_{\tau|t} \coloneqq \phi_{\tau}^{\star}-\phi_{\tau|t}, \ \    \overline{\epsilon}_{\tau|t} \coloneqq \phi_{\tau}^{\star}-\kappa_{\tau|t}
\end{align}
for all $\tau\in [t:\overline{t}]$. Stacking the step-wise prediction errors leads to the following \textit{overall errors} for the predictive and nominal parameters $\phi_{\textsc{P}}\coloneqq (\phi_{t:\overline{t}|t}:t\in [\nt])$ and $\kappa_{\textsc{N}}\coloneqq (\kappa_{t:\overline{t}|t}:t\in [\nt])$ with respect to the true parameters $\phi_{\star}\coloneqq (\phi^{\star}_{t}:t\in [\nt])$:
\begin{align}
\nonumber
% \boldsymbol{\varepsilon}\left(\phi_{\textsc{P}},\phi_{\star}\right) &\coloneqq \sum_{t=0}^{\nt-1}\Big(\sum_{\tau=t}^{\overline{t}}{\varepsilon}_{\tau |t}\Big)^2,  \ \ 
% \overline{\boldsymbol{\varepsilon}}\left(\phi_{\textsc{N}},\phi_{\star}\right) \coloneqq \sum_{t=0}^{\nt-1}\Big(\sum_{\tau=t}^{\overline{t}}\overline{\varepsilon}_{\tau |t}\Big)^2.
\boldsymbol{\varepsilon}\left(\phi_{\textsc{P}},\phi_{\star}\right) &\coloneqq \left({\varepsilon}_{\tau |t}\right)_{\tau\in [t,\overline{t}],t\in [\nt]},  \ \ 
\overline{\boldsymbol{\varepsilon}}\left(\phi_{\textsc{N}},\phi_{\star}\right) \coloneqq \left(\overline{\varepsilon}_{\tau |t}\right)_{\tau\in [t,\overline{t}],t\in [\nt]}.
\end{align}
We will simplify them as $\boldsymbol{\varepsilon}$ and $\overline{\boldsymbol{\varepsilon}}$ if they are regarded as static quantities rather than mappings.
In this work, we focus on designing a competitive policy (c.f.~\cite{goel2022competitive}) $\pi=(\pi_t:t\in [\nt])$ that minimizes the competitive ratio of costs as a function of the overall errors $\boldsymbol{\varepsilon}$ and $\overline{\boldsymbol{\varepsilon}}$. 
We write $\Phi^\nk$ as a $\nk$-ary Cartesian product of the uncertainty sets $\Phi$'s. At time $t\in [\nt]$, each control policy $\pi_t:\mathbb{R}^{\nn}\times \Phi^{\nk}\rightarrow \mathbb{R}^{\nm}$ is a state-feedback policy that produces an action $u_t\in \mathbb{R}^{\nm}$ with an observed state $x_{t-1}$, and predictions $\phi_{t:\overline{t}|t}$ of the variables $\phi_{t:\overline{t}}^{\star}\in \Phi^{\nk}$ for future $\nk$ time steps with a prediction window size $\nk >0$. Formally, we define
\begin{align}
    \label{eq:cr}
\mathsf{CR}\left(\pi;\boldsymbol{\varepsilon},\overline{\boldsymbol{\varepsilon}}\right)\coloneqq \sup_{\phi_{\textsc{P}},\phi_{\textsc{N}},\phi_{\star}\in\Pi}\frac{J(\pi)}{J^{\star}}
\end{align}
where the supremum is taken over $\Pi\coloneqq \Phi^{\nk\times\nt}\times \Phi^{\nk\times\nt} \times \Phi^{\nt}$; $J(\pi)$ is the total cost induced by $\pi$ as defined in~\eqref{eq:costs} subject to the constraints in~\eqref{eq:dynamics} and~\eqref{eq:constraints}, and $J^{\star}$ denotes the offline optimal cost. When $J^{\star}=0$, we let $\mathsf{CR}\left(\pi;\boldsymbol{\varepsilon},\overline{\boldsymbol{\varepsilon}}\right)=\infty$. 

Our goal is to design $\pi$ and solve the Learning-Augmented Control (LAC) problem such that (1). $\pi$
exploits good predictions \textit{consistently} to achieve near-optimal competitive ratio, i.e., $\limsup_{\nt\rightarrow\infty}\mathsf{CR}\left(\pi;\boldsymbol{\varepsilon},\overline{\boldsymbol{\varepsilon}}\right)=1$ when $\boldsymbol{\varepsilon}=0$; (2). $\pi$
has a bounded competitive ratio \textit{robustly} in the worst-case, i.e., $\limsup_{\nt\rightarrow\infty}\sup_{\boldsymbol{\varepsilon}}\mathsf{CR}\left(\pi;\boldsymbol{\varepsilon},\overline{\boldsymbol{\varepsilon}}\right)=O(1)$.

\section{Learning-Augmented Control}

In this section, based on the MPC scheme defined in Section~\ref{sec:model_mpc}, we explore solutions to achieve the control target of the LAC problem.
A natural solution to the general problem~\eqref{eq:dynamics}-\eqref{eq:costs} is to use an MPC scheme with predictive and nominal parameters $\phi_{t:\overline{t}}$ and $\kappa_{t:\overline{t}}$. 
Since the predictive parameters $\phi_{t:\overline{t}}$ may contain error, at each time $t$ we combine predictive parameters $\phi_{t:\overline{t}}$ and  nominal parameters $\kappa_{t:\overline{t}}$~in~\eqref{eq:mpc1}-\eqref{eq:mpc2} using a  \textit{confidence parameter} $\lambda\in \mathcal{I}$ such that
\begin{align}
\label{eq:lcc}
u_t\left(x_t;\lambda \phi_{t:\overline{t}|t}+(1-\lambda) \kappa_{t:\overline{t}|t}\right) \ \text{($\lambda$-confident control)}
\end{align}
is taken.
Fixing $\lambda\in\mathcal{I}$, we denote this \textit{$\lambda$-confident control} in~\eqref{eq:lcc} as $\pi_{\lambda}:\mathbb{R}^{\nn}\rightarrow\mathbb{R}^{\nm}$. We  validate its recursive feasibility.

\begin{theorem}
\label{thm:recursive_feasibility}
Suppose Assumption~\ref{ass:control_invariant} holds.
There always exist MPC solutions when
implementing the $\lambda$-confident control policy $\pi_{\lambda}$. Furthermore, the recursive feasibility holds if a sequence of time-varying $(\lambda_t:t\in [\nt])$ is used.
\end{theorem}

\begin{proof}
The uncertainty set $\Phi$ is convex, thus for any $\tau\in [t:\overline{t}]$, $t\in [\nt]$, and $\lambda_t\in\mathcal{I}$, $\phi_{\tau|t}^{\lambda_t}\coloneqq\lambda_t \phi_{\tau|t}+(1-\lambda_t) \kappa_{\tau|t}\in\Phi$. The initial state $x_0\in\mathsf{X}$ implies that the MPC problem~\eqref{eq:mpc} is feasible at time $t=0$ since $\mathsf{X}$ is control invariant. We prove the theorem using induction by showing that if the system is feasible at time $t$ with a starting state $x_t\in\mathsf{X}$, then it is also feasible at time $t+1\in [\nt]$ with an actual state $x_{t+1}\in\mathsf{X}$ at $t+1$. Since by the inductive hypothesis, $x_t\in\mathsf{X}$, there exists a feasible action $u_{t}\in\mathsf{U}$ such that $f_t(x_t,u_t;\phi_{t|t}^{\lambda_t})\in\mathsf{X}$ and $h_t(x_t,u_t;\phi_{t|t}^{\lambda_t})\leq 0$. This action $u_t$ is the actual action taken at time $t$. Therefore, the robustness of the uncertainty set $\Phi$ (Assumption~\ref{ass:control_invariant}) implies that the actual state at time $t+1$  satisfies $f_t(x_t,u_t;\phi_{t}^{\star})\in\mathsf{X}$ and $h_t(x_t,u_t;\phi_{t}^{\star})\leq 0$ with the same $(x_t,u_t)$. Similarly, the constraint~\eqref{eq:mpc4} guarantees the feasibility of the MPC scheme at time $t+1$.
\end{proof}

\subsection{Sensitivity Analysis}

% \begin{definition}[$\rho$-Weighted Inner Product]
% \label{def:weighted_inner_product}

% \end{definition}

% Note that the $\rho$-weighted norm can be bounded from above by the $\ell_2$-norm such that $\left\|\boldsymbol{a}\right\|_{\rho}\leq C \left\|\boldsymbol{a}\right\|$, by the definition of $\rho$ in~\Cref{def:perturbation_bound}. 

Recall the mapping $\rho:\mathbb{N}\rightarrow\mathbb{R}_+$ defined in~\Cref{sec:problem_formulation}. 
To deliver our sensitivity analysis result and enable the online learning of $\lambda$, we consider a \textit{total error} that quantifies how $\lambda$ affects the system performance, represented as a $\rho$-weighted $\ell_2$-norm 
${\xi}(\lambda)\coloneqq \|\lambda\boldsymbol{\varepsilon}+(1-\lambda)\overline{\boldsymbol{\varepsilon}}\|_{\rho}^2$ of a convex combination of $\boldsymbol{\varepsilon}$ and $\overline{\boldsymbol{\varepsilon}}$, where $\lambda$ is the confidence parameter in the $\lambda$-confident control. The $\rho$-weighted $\ell_2$-norm 
is defined as 
${\xi}(\lambda) = \sum_{t\in [\nt]}\xi_{t,\nt}(\lambda)$, with each $\xi_{t,\nt}(\lambda)$ defined as
\begin{align}
\label{eq:xi_t}
% \big(\sum_{\tau=t}^{\overline{t}}\rho(\tau-t)\varepsilon_{\tau |t}\big)^2
%  + (1-\lambda)^2\big(\sum_{\tau=t}^{\overline{t}}\rho(\tau-t)\overline{\varepsilon}_{\tau |t}\big)^2 
\xi_{t,\nt}(\lambda)\coloneqq \left(\sum_{\tau=t}^{\overline{t}}\rho(\tau-t)\left\|\lambda \varepsilon_{\tau|t} +(1-\lambda)\overline{\varepsilon}_{\tau|t}\right\|\right)^2.
\end{align}
For notional simplicity, we further let $\xi_{\tau|t}(\lambda)\coloneqq \rho(\tau-t)\left\|\lambda \varepsilon_{\tau|t} +(1-\lambda)\overline{\varepsilon}_{\tau|t}\right\|$.
We analyze the cost difference corresponding to the $\lambda$-confident control and the optimal cost, in terms of the step-wise prediction errors~\eqref{eq:prediction_error}, as shown in the following lemma, whose proof is provided in~\Cref{app:proof_lemma_regret_error}.
\begin{lemma}
\label{lemma:regret_error}
Suppose the bounds in~\Cref{def:perturbation_bound} (EDPBs) hold and the $\lambda$-confident control policy is feasible.
The $\lambda$-confident control in~\eqref{eq:lcc} satisfies for any $\eta>0$,
    \begin{align*}
      J(\pi_{\lambda}) - (1+\eta)J^{\star} 
      \leq \left(3C^2+2\right)\left(\ell+\frac{\ell}{\eta}\right)\left(  \gamma^2\rho^2(\nk)\nt + {\xi}(\lambda)
\right)
%       C\left(L_c+2L_c L_f\right)\left(  \nt\rho(\nk) + 
% \sum_{t=0}^{\nt-1}\xi_{t}(\lambda)\right),
\end{align*}
where $C$ is a constant defined in~\Cref{sec:second_first_assumption}. 
\end{lemma}
It is well known that for a fixed $\lambda$, the $\lambda$-confident control defined in~\eqref{eq:lcc} may not yield a stabilizing controller (see Theorem 1 in~\cite{li2023certifying} and Section II of~\cite{zheng2020equivalence}). Consequently, we seek to learn a sequence of time-varying parameters $(\lambda_t:t\in [\nt])$ online, a process we term \emph{confidence learning}. In what follows, we present a formal scheme for the online estimation of these confidence parameters.

% \begin{figure*}
% \begin{align}
% \label{eq:online_function}
% \zeta_t(\lambda)\coloneqq 2(k+1)\sum_{\tau=\underline{t}}^{t} \underbrace{\Bigg(\rho^2(t-\tau) \Big(\lambda^2\varepsilon^2_{t|\tau}
%  + (1-\lambda)^2  \overline{\varepsilon}^2_{t|\tau}\Big)\Bigg)}_{\zeta_{t,\tau}(\lambda)}
% \end{align}
% \end{figure*} 

\subsection{Delayed Confidence Learning}

To minimize the regret $J(\pi_{\lambda}) - J^{\star}$, we aim to minimize its upper bound in~\Cref{lemma:regret_error} by generating the confidence parameter $\lambda_t$ online at each time $t$. 
We proceed to reformulate the minimization of the upper bound in~\Cref{lemma:regret_error} into an online optimization problem. Consider the mapping $\xi_{t,\nt}(\lambda):\mathcal{I}\rightarrow\mathbb{R}$, which is defined as a sum over $\tau$ for $t \leq \tau \leq \overline{t}$. The structure of the offline objective $\sum_{t=0}^{\nt-1}\xi_{t,\nt}(\lambda)$ (see~\eqref{eq:xi_t}) depends on the full problem horizon $\nt$, making it incompatible with standard online optimization frameworks.  
This necessitates a tailored approach for learning $\lambda_t$ sequentially, which is detailed next.

\begin{figure}[ht]
    \centering    \includegraphics[width=0.8\linewidth]{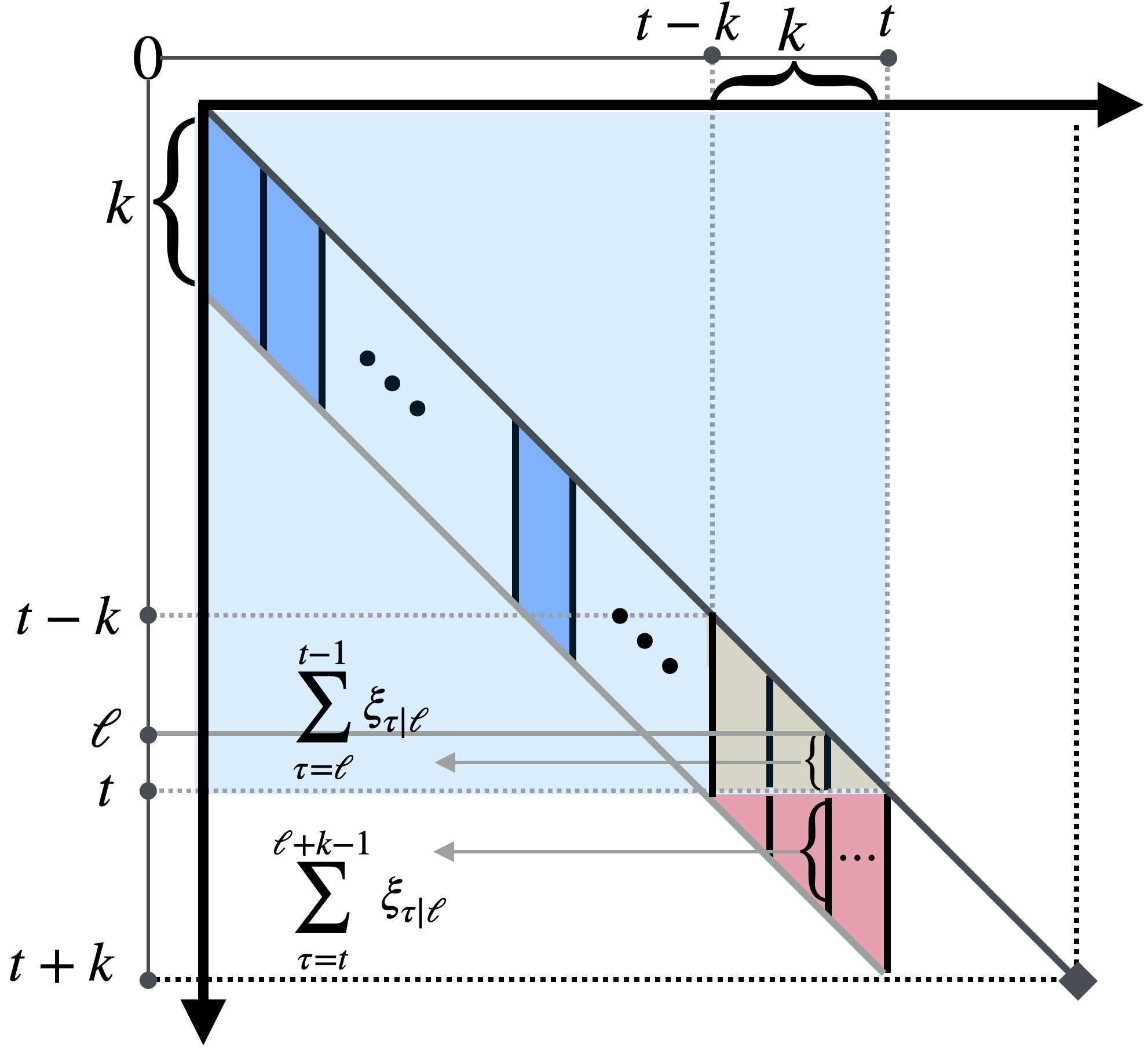}
    \caption{Decomposition of the term $\smash{\xi_{\ell,t}^{1/2}}$ as a sum of $k$ elements $\xi_{\tau|\ell}$ into known information ($\sum_{\tau=\ell}^{t-1}\xi_{\tau|\ell}$) and future information ($\sum_{\tau=t}^{\ell+\nk-1}\xi_{\tau|\ell}$) with respect to the decision time $t$. The light blue area indicates information available at time $t$. This dependency on future components necessitates the \dcl online optimization procedure (Procedure~\ref{alg:dcl}).}
    \label{fig:enter-label}
\end{figure}

In particular, the terms
in this upper bound depend on information that unfolds over time. At time $t$, consider the terms $(\xi_{\ell,t}(\lambda):\ell< t)$.
For $t> \ell> t-\nk$, each $\xi_{\ell,t}(\lambda)$  can be decomposed based on information available at the current decision time $t\in [\nt]$ into known and future components. Figure~\ref{fig:enter-label} visually illustrates this decomposition for a generic term $\xi_{\ell,t}(\lambda)$ and shows how $\xi_{\ell,t}(\lambda)$ is constructed from elements $\xi_{\tau|\ell}(\lambda)$ (summed over a window of size $\nk$, which corresponds to $\nk$ in the formula below) and how these elements are partitioned relative to the current decision time $t\in [\nt]$.
\begin{align*}
    \xi_{\ell,t}(\lambda) = \Big(\underbrace{\sum_{\tau=\ell}^{t-1}\xi_{\tau | \ell}(\lambda)}_{\text{known information}} + \underbrace{\sum_{\tau=t}^{\min\{\ell+\nk-1,\nt-1\}}\xi_{\tau | \ell}(\lambda)}_{\text{future information}}\Big)^2.
\end{align*}
At time $t$, before determining the control action $u_t$, the controller receives delayed feedback consisting of $(\xi_{\ell,t}:0\leq\ell\leq t-\nk)$.
Since each term $\xi_{\ell,t}$
for $\ell> t-\nk$ incorporates ``future information'' not yet observed due to this delay, optimizing the upper bound using only the available delayed feedback naturally leads to a delayed update rule for $\lambda_t$, using the gradient of $\xi_{t-\nk}(\lambda)$. Here, we write it as $\xi_{t-\nk}(\lambda)$ instead of $\xi_{t-\nk,\nt}(\lambda)$ since $t\leq \nt-1$. This procedure, denoted \dcl, is outlined in Procedure~\ref{alg:dcl} and is consistent with established delayed online optimization techniques~\cite{weinberger2002delayed,joulani2013online}. It is worth noting that alternative delayed online optimization schemes exist that achieve improved multiplicative constants~\cite{huang2023banker,van2023unified}.

\begin{myprocedure}[t]

\textbf{Input}: Parameters $\phi^{\star}_{t-\nk:t-1}$, $\phi_{t-\nk:t-1|t-\nk}$, and $\kappa_{t-\nk:t-1|t-\nk}$

\SetAlgoLined

\DontPrintSemicolon

Compute ${\varepsilon}_{\tau |t-\nk}$ and $\overline{\varepsilon}_{\tau |t-\nk}$ as~\eqref{eq:prediction_error} for $t-\nk\leq \tau\leq t-1$

Update $\lambda_t = \mathsf{Proj}_{\mathcal{I}}\left(\lambda_{t-\nk} - {\beta}\nabla_{\lambda}\xi_{t-\nk}(\lambda)\right)$

\Return $\lambda_t$

\caption{\textsc{D}elayed  \textsc{C}onfidence \textsc{L}earning (\dcl)}
\label{alg:dcl}
\end{myprocedure}

% \begin{subequations}
% \begin{align}
% \min_{u_{\underline{t}:t}} & \sum_{\tau=t}^{\overline{t}} c(x_\tau,u_\tau;\phi_\tau) + C(x_{\overline{t}+1};\kappa_{\overline{t}+1})\\
%     \text{s.t. } & x_{\tau + 1} = f(x_{\tau},u_{\tau},\phi_{\tau}), \ \tau=t,\ldots,\overline{t}
% \end{align}
% \end{subequations}

The following lemma establishes a regret bound for the delayed online optimization of the confidence parameter $\lambda_t$.

\begin{lemma}
\label{lemma:dcl}
Suppose the bounds specified in~\Cref{def:perturbation_bound} (EDPBs) hold.
Let $\lambda^{\star}$ be a parameter that minimizes the cumulative term $\sum_{t=0}^{\nt-1}\xi_{t,\nt}(\lambda)$ over the feasible set $\in\mathcal{I}$. Then, the sequence of parameters $\{\lambda_t\}_{t=0}^{\nt-1}$ generated by the \dcl procedure (Procedure~\ref{alg:dcl}) satisfies:
\begin{align}
\label{eq:dcl_regret}
\sum_{t=0}^{\nt-1}\xi_{t,\nt}\left(\lambda_t\right) -  \sum_{t=0}^{\nt-1}\xi_{t,\nt}\left(\lambda^{\star}\right)\leq 4C^2\gamma^2\sqrt{\nt\nk+\nk^2}.
\end{align} 
\end{lemma}

The proof is provided in~\Cref{app:proof_lemma_dcl_regret}.
% It is important to note that the term $\xi_{t,\nt}$ in~\eqref{eq:xi_t} does not depend on $\rho$, as for general systems, the mapping $\rho$ is often unknown. For the case when $\rho$ is given, the bound in~\eqref{eq:dcl_regret} can be improved to $4\gamma^2\sqrt{\nt\nk+\nk^2}$ by improving the gradient bound in~\eqref{eq:proof_gradient} to $2C^2\gamma^2$.
For general systems, the mapping $\rho$ is often unknown. In this case, the gradient bound  in~\eqref{eq:xi_t} can still be bounded by $2C^2\gamma^2\nk^2$,
then the regret bound in~\eqref{eq:dcl_regret} becomes $4C^2\gamma^2\sqrt{\nt\nk^3+\nk^4}$, which is still sublinear in $\nt$ as long as the prediction window size satisfies $\nk=\Theta(\log\nt)$. Thus, in the following context, we assume $\rho$ is given for simplicity.

\subsection{Algorithm Description}

Together, the implementation of $\lambda_t$-confident control $\pi_{\lambda_t}$ for $t\in [\nt]$ with \dcl is termed \texttt{L}earning-\texttt{A}ugmented \texttt{C}ontrol (\ouralg), as summarized in Algorithm~\ref{alg:lac}.

\begin{algorithm}[ht]

\SetAlgoLined
\For{$t=0,\ldots,\nt-1$}{
\DontPrintSemicolon
\tcc{\color{black} Learning confidence $\lambda_t$ online}

\lIf{$t<\nk$}{Initialize arbitrary $\lambda_t\in \mathcal{I}$}
\lElse{
Obtain last ground-truth parameter $\phi^{\star}_{t-1}$

\quad Set $\lambda_t = \dcl\left(\phi^{\star}_{t-\nk}, \phi_{t-\nk:t-1|t-\nk}, \kappa_{t-\nk:t-1|t-\nk}\right)$
}

\quad Obtain predictions $(\phi_{\tau|t}: t\leq \tau\leq \overline{t})$

% \quad Compute \texttt{MPC} solution $u_t(x_t;\phi_{t:\overline{t}})$ 

% based on predictions and nominal values

\quad \tcc{\color{black} Run $\lambda_t$-confident control $\pi_{\lambda_t}$}

\quad Take action $u_t\left(x_t;\lambda \phi_{t:\overline{t}}+(1-\lambda) \kappa_{t:\overline{t}}\right)$ 
 
\quad Update and observe $x_{t+1}$ in~\eqref{eq:dynamics}
}
\caption{\textsc{L}earning-\textsc{A}ugmented Control (\ouralg)}
\label{alg:lac}
\end{algorithm}

The algorithm iteratively performs two main steps at each time $t$: online learning of a confidence parameter $\lambda_t$, and the execution of a $\lambda_t$-confident control action.
For the learning phase, if the current time $t$ is less than the delay parameter $\nk$ (i.e., $t < \nk$), $\lambda_t$ is initialized arbitrarily within a feasible set $\mathcal{I}$. Otherwise, for $t \geq \nk$, $\lambda_t$ is updated using the \dcl\ procedure. As indicated in Algorithm~\ref{alg:lac}, this update utilizes the ground-truth parameter $\phi^{\star}_{t-\nk}$ (observed $\nk$ steps prior), historical predictions $\phi_{t-\nk:t-1|t-\nk}$, and nominal values $\kappa_{t-\nk:t-1|t-\nk}$.
% Note: The algorithm states "Obtain last ground-truth parameter \phi^{\star}_{t-1}" but uses \phi^{\star}_{t-\nk} in the DCL call. The description reflects the DCL call's arguments.

Following the determination of $\lambda_t$, the algorithm obtains current predictions, denoted as $(\phi_{\tau|t} : t\leq \tau\leq \bar{t})$. Based on these predictions and corresponding nominal values (represented by $\kappa_{t:\overline{t}}$), the final control action $u_t$ is an MPC solution with a convex combination of predictive and nominal parameters, weighted by $\lambda_t$, aligning with the $\lambda$-confident control in~\eqref{eq:lcc}.
% \[
% u_t=\lambda_t u_t(x_t;\phi_{t:\overline{t}})+(1-\lambda_t) u_t(x_t;\kappa_{t:\overline{t}}).
% \]
This action $u_t$ is then applied to the system, and the next state $x_{t+1}$ is observed according to the system dynamics in~\eqref{eq:dynamics}. Finally, we denote the \ouralg policy as $\pi_{\ouralg}$.

\section{Main Results}
\label{sec:main}

We summarize the theoretical results of \ouralg in this section. First, we present the best-of-both-world guarantee in terms of the competitive ratio. Second, we revisit classic linear quadratic control (LQC) setting and show that our general results improve the existing bound in~\cite{li2022robustness}.

\subsection{Competitive Control Guarantee}
Drawing parallels with the style of guarantees in the competitive control literature~\cite{goel2022competitive,yu2022competitive,li2022robustness}, we establish the following upper bound on the competitive ratio $\mathsf{CR}\left(\pi_{\ouralg};\boldsymbol{\varepsilon},\overline{\boldsymbol{\varepsilon}}\right)$ of \ouralg, as defined in~\eqref{eq:cr}. Recall the $\rho$-weighted norm is defined as $\|\boldsymbol{\varepsilon}\|_\rho^2\coloneqq \sum_{t\in [\nt]}(\sum_{\tau=t}^{\overline{t}}\rho(\tau-t)\left\|\varepsilon_{\tau|t}\right\|)^2$ for a stacked vector $\boldsymbol{\varepsilon}$ and $\|\overline{\boldsymbol{\varepsilon}}\|_\rho^2$ is defined similarly.

\begin{theorem}
\label{thm:lac}
Under Assumption~\ref{ass:basic},~\ref{ass:control_invariant},~\ref{ass:mpc} and~\ref{ass:singular_bounds},
the \ouralg policy $\pi_{\ouralg}$ satisfies
$
\mathsf{CR}(\pi_{\ouralg})\leq 1+2\sqrt{\zeta}+\zeta 
$ for some constant $C>0$ and a mapping $\rho:\mathbb{R}\rightarrow\mathbb{R}$ satisfying $\sum_{t\in [\nt]}\rho(t)\leq C$
where 
\begin{align}
\nonumber
&\zeta\coloneqq \frac{3C^2+2}{J^{\star}}\ell\Big(\gamma^2\rho^2(\nk)\nt  + 
\varpi\left(\boldsymbol{\varepsilon},\overline{\boldsymbol{\varepsilon}}\right)
+4C^2\gamma^2\sqrt{\nt\nk+\nk^2}\Big),\\
\label{eq:def_zeta}
&\varpi\left(\boldsymbol{\varepsilon},\overline{\boldsymbol{\varepsilon}}\right)\coloneqq \frac{\|\boldsymbol{\varepsilon}\|_\rho^2 \|\overline{\boldsymbol{\varepsilon}}\|_\rho^2}{\|\boldsymbol{\varepsilon}\|_\rho^2 + \|\overline{\boldsymbol{\varepsilon}}\|_\rho^2}.
\end{align}
\end{theorem}

Theorem~\ref{thm:lac} provides a non-asymptotic bound on the competitive ratio, showing that it approaches the ideal value of~$1$ as the term $\zeta$ diminishes. The expression for $\zeta$ in~\eqref{eq:def_zeta} is particularly insightful as it decomposes the performance loss into three primary sources. The first term, involving $\nt\rho^2(\nk)$, captures the cost of finite-horizon planning, which fades as the influence of distant stages, represented by the decaying function $\rho(\cdot)$, becomes negligible. The second term of the overall prediction errors $\overline{\boldsymbol{\varepsilon}}$ and $\boldsymbol{\varepsilon}$ quantifies the impact of both the nominal and untrusted predictive parameters (see Section~\ref{sec:problem_formulation_error}). The final sublinear regret term, scaled by $\gamma^2$ (where $\gamma>0$ is the diameter of the uncertainty set $\Phi$), represents the price of robustness against adversarial disturbances.

To better illustrate the asymptotic behavior of this bound, we can analyze it in the common regime where the optimal cost is significant, i.e., $J^{\star}=\Omega(\nt)$, and the prediction window is chosen to scale logarithmically with the horizon, i.e., $\nk=\Theta(\log\nt)$. In this practical setting, the bound in Theorem~\ref{thm:lac} elegantly simplifies to:
\begin{align*}
 \mathsf{CR}(\pi_{\ouralg})\leq 1 +O\left(\frac{1}{\nt}\left(\left(\varpi\left(\boldsymbol{\varepsilon},\overline{\boldsymbol{\varepsilon}}\right)\right)^{1/2}+\varpi\left(\boldsymbol{\varepsilon},\overline{\boldsymbol{\varepsilon}}\right)\right)\right).
\end{align*}
This result clearly demonstrates that as prediction accuracy improves ($\boldsymbol{\varepsilon} \to \boldsymbol{0}$), the competitive ratio of \ouralg converges to $1$, i.e., $\limsup_{\nt\rightarrow\infty}\mathsf{CR}\left(\pi;\boldsymbol{0},\overline{\boldsymbol{\varepsilon}}\right)=1$; otherwise, \ouralg has a bounded competitive ratio robustly in the worst-case, since $\sup_{\boldsymbol{\varepsilon}}\varpi\left(\boldsymbol{\varepsilon},\overline{\boldsymbol{\varepsilon}}\right) \leq \|\overline{\boldsymbol{\varepsilon}}\|_{\rho}^2$ implies that $\limsup_{\nt\rightarrow\infty}\sup_{\boldsymbol{\varepsilon}}\mathsf{CR}\left(\pi;\boldsymbol{\varepsilon},\overline{\boldsymbol{\varepsilon}}\right)=1+ O(\|\overline{\boldsymbol{\varepsilon}}\|_{\rho}^2/\nt) = 1 + O(C^2\gamma^2) = O(1)$.

% In~\eqref{eq:def_zeta}, recall that $\ell$ is defined in Assumption~\ref{ass:basic}; $\gamma>0$ bounds the radius of the convex uncertainty set $\Phi$; $\boldsymbol{\varepsilon}$ and $\overline{\boldsymbol{\varepsilon}}$ are overall prediction errors corresponding to predictive and nominal parameters (see Section~\ref{sec:problem_formulation_error}). In particular, in the common regime of $J^{\star}=\Omega(\nt)$ (e.g., costs are strictly positive), with a prediction window size satisfying $\nk=\Theta(\log\nt)$, Theorem~\ref{thm:lac} implies
% \begin{align*}
%  \mathsf{CR}(\pi_{\ouralg})\leq 1 +O\left(\left(\frac{\overline{\boldsymbol{\varepsilon}}\boldsymbol{\varepsilon}}{\overline{\boldsymbol{\varepsilon}}+\boldsymbol{\varepsilon}}\right)^{1/2}+\frac{\overline{\boldsymbol{\varepsilon}}\boldsymbol{\varepsilon}}{\overline{\boldsymbol{\varepsilon}}+\boldsymbol{\varepsilon}}\right).
% \end{align*}

\begin{proof}[Proof of~\Cref{thm:lac}]
The proof begins by confirming from Theorem~\ref{thm:recursive_feasibility} that the policy generated by \ouralg is always feasible. Building on this, we leverage Lemma~\ref{lemma:regret_error} to establish an initial upper bound on the competitive ratio in terms of a general error function ${\xi}(\lambda)$ by optimizing over $\eta>0$:
\begin{align}
\nonumber
\frac{J\left(\pi\right)}{J^{\star}}
\leq 1 + 2\sqrt{\frac{3C^2+2}{J^{\star}}\ell\left(\gamma^2\rho^2(\nk)\nt+{\xi}(\lambda)\right)}&\\
+\frac{3C^2+2}{J^{\star}}\ell\left(\gamma^2\rho^2(\nk)\nt+{\xi}(\lambda)\right)&.
\end{align}
The next step is to bound the error term ${\xi}(\lambda)$.
% Recall $\xi_{t,\nt}(\lambda)\coloneqq\sum_{\tau=t}^{\overline{t}}\left\|\lambda \varepsilon_{\tau|t} +(1-\lambda)\overline{\varepsilon}_{\tau|t}\right\|^2$. 
 By combining the inequality in~\Cref{lemma:regret_error} with the result of Lemma~\ref{lemma:dcl}, we can relate this error to its minimum possible value plus a sublinear regret bound that depends on the system's uncertainty diameter $\gamma$:
\begin{align*}
{\xi}(\lambda)\leq  
\sum_{t=0}^{\nt-1}{\xi}_{t,\nt}\left(\lambda^{\star}\right)  + 4C^2\gamma^2\sqrt{\nt\nk+\nk^2},
\end{align*}
where, by its definition, the minimal error $\sum_{t=0}^{\nt-1}{\xi}_{t,\nt}\left(\lambda^{\star}\right)$ can be calculated as the least squares:
\begin{align*}
    \sum_{t=0}^{\nt-1}{\xi}_{t,\nt}\left(\lambda^{\star}\right) = & \min_{\lambda\in\mathcal{I}}\sum_{t=0}^{\nt-1}{\xi}_{t,\nt}\left(\lambda\right)
    =\frac{\|\boldsymbol{\varepsilon}\|_\rho^2 \|\overline{\boldsymbol{\varepsilon}}\|_\rho^2}{\|\boldsymbol{\varepsilon}\|_\rho^2 + \|\overline{\boldsymbol{\varepsilon}}\|_\rho^2}.
\end{align*}
Substituting these expressions back into our initial inequality confirms the relationship $\frac{J\left(\pi\right)}{J^{\star}}
\leq 1 +2\sqrt{\zeta} + \zeta$, where $\zeta$ is precisely the term defined in~\eqref{eq:def_zeta}.
\end{proof}

% \begin{align}
% \nonumber
%       J(\pi_{\lambda}) - J^{\star} 
%       \leq & C\left(L_c+2L_c L_f\right)\left(  \nt\rho(\nk) + 
% \sum_{t=0}^{\nt-1}\left(\xi_{t}(\lambda) + \zeta_{t}(\lambda)\right)\right),
% \end{align}
% where 
% \begin{align*}
%     \zeta_{t}(\lambda)\coloneqq
% \end{align*}

% \begin{proof}
% By Theorem~\ref{thm:recursive_feasibility}, \ouralg is feasible. 
% Lemma~\ref{lemma:regret_error} implies that
% \begin{align}
% \nonumber
% \frac{J\left(\pi\right)}{J^{\star}}
% \leq 1 + 2\sqrt{\frac{6C^2+2}{J^{\star}}\ell\left(\nt\rho^2(\nk)+{\xi}(\lambda)\right)}&\\
% +\frac{6C^2+2}{J^{\star}}\ell\left(\nt\rho^2(\nk)+{\xi}(\lambda)\right)&.
% \end{align}

% % \begin{align*}
% % J(\pi_{\lambda}) - J^{\star} 
% %       \leq C\left(L_c+2L_c L_f\right)\left(  \nt\rho(\nk) + 
% % \sum_{t=0}^{\nt-1}\xi_{t}(\lambda)\right),
% % \end{align*}

% Furthermore,~\eqref{eq:upper_bound_noncanonical} and Lemma~\ref{lemma:dcl} imply
% \begin{align*}
% {\xi}(\lambda)\leq  
% \sum_{t=0}^{\nt-1}\xi_{t}\left(\lambda^{\star}\right)  + 8C^4\gamma^4\sqrt{\nt\nk^3+\nk^4},
% \end{align*}
% where by definition, it follows that
% \begin{align*}
%     \sum_{t=0}^{\nt-1}\xi_{t}\left(\lambda^{\star}\right) = & \min_{\lambda\in\mathcal{I}}\sum_{t=0}^{\nt-1}\xi_{t}\left(\lambda\right)
%     =\frac{\overline{\boldsymbol{\varepsilon}}\boldsymbol{\varepsilon}}{\overline{\boldsymbol{\varepsilon}}+\boldsymbol{\varepsilon}}.
% \end{align*}

% As a result, we conclude that $\frac{J\left(\pi\right)}{J^{\star}}
% \leq 1 +2\sqrt{\zeta} + \zeta$
% where $\zeta$ is defined in~\eqref{eq:def_zeta}.
    
% \end{proof}

\subsection{Revisit LQC}

Now, we turn to consider a realization of Theorem~\ref{thm:lac} for the LQC problem, introduced in Section~\ref{sec:examples}. Aligning with the $\lambda$-confident control policy in~\cite{li2022robustness}, for LQC, the MPC scheme~\eqref{eq:mpc} boils down to
\begin{subequations}
\label{eq:lqc_problem}
\begin{align}
\label{eq:lqc_problem1}
\min_{u_{t:\overline{t}},x_{t+1:t'}}  \sum_{t=0}^{\nt-1} \left(x_t^\top Q x_t + u_t^\top R u_t\right) & + x_{\nt}^\top P x_{\nt}\\
\label{eq:lqc_problem2}
    \text{s.t. }  x_{\tau + 1} = Ax_{\tau}+Bu_{\tau}+\phi_{\tau|t}^{\lambda}, \ & \tau=t,\ldots,\overline{t},\\
\label{eq:lqc_problem3}
x_{t'}\in{\mathbb{R}}, \ x_\tau\in\mathbb{R}, u_{\tau}\in\mathbb{R}, \ & \tau=t,\ldots,\overline{t},
\end{align}
\end{subequations}
where $\phi_{\tau|t}^{\lambda}=\lambda \phi_{\tau|t}$ by setting all nominal predictions $(\kappa_{t:\overline{t}}:t\in [\nt])$ to be zeros.
It is well known that the optimal solution of~\eqref{eq:lqc_problem} can be explicitly written as (see~\cite{yu2020power,li2022robustness})
\begin{align}
\nonumber
u_t^{\mathsf{LQC}}(x_t;\phi_{t:\overline{t}})=-Kx_t -\lambda(R+B^{\top}PB)^{-1}B^{\top}\sum_{\tau=t}^{\overline{t}}(F^{\top})^{\tau-t}P\phi_{\tau|t},
\end{align}
as a concrete example of~\eqref{eq:lcc}. We define the \textit{adversity} of this LQC problem as the cardinality of the set of nonzero disturbances, e.g., $\mathtt{Ad}(\phi_{\star})\coloneqq |\{t\in [\nt]:\|\phi_t^{\star}\|>0\}|$. Let $\boldsymbol{\phi}^{\star}$ stack the terms $\phi_{\tau}^{\star}$ over $\tau\in [t,\overline{t}]$ and $t\in [\nt]$ and let $\varpi\left(\boldsymbol{\varepsilon},\boldsymbol{\phi}^{\star}\right)\coloneqq \frac{\|\boldsymbol{\varepsilon}\|^2\|\boldsymbol{\phi}^{\star}\|^2-\left\langle\boldsymbol{\varepsilon},\boldsymbol{\phi}^{\star}\right\rangle^2}{\|\boldsymbol{\varepsilon}-\boldsymbol{\phi}^{\star}\|^2}$ as a least-squares minimizer (whose numerator is the Gram determinant).
The following holds:
\begin{theorem}
\label{thm:lac_lqc}
For the LQC problem,
the \ouralg policy satisfies
\begin{align*}
\mathsf{CR}(\pi_{\ouralg})\leq  1&+\frac{1}{J^{\star}}\frac{2C_F^2 \|H\|\|P\|^2}{\left(1-\rho_F\right)^2}\Big( 
\gamma^2\rho_{F}^{2\nk}\nt+ 4\gamma^2\sqrt{\nt\nk^3+\nk^4}\\
& + \varpi\left(\boldsymbol{\varepsilon},\boldsymbol{\phi}^{\star}\right)
\Big).
\end{align*}
Furthermore, if $\mathtt{Ad}(\phi_{\star})=\omega(\sqrt{\nt\log \nt})$, then
$
\mathsf{CR}(\pi_{\ouralg})\leq 1+ O\left(\|{\boldsymbol{\varepsilon}}\|/(\|{\boldsymbol{\varepsilon}\|+\sqrt{\nt\log\nt}})\right) + o(1)
$
with $\nk=O(\log\nt)$.
\end{theorem}

\begin{proof}

For LQC, note the following explicit form of the overall dynamic regret (see~\cite{yu2022competitive,li2022robustness}) with a fixed $\lambda\in\mathcal{I}$:
\begin{align}
\label{eq:overall_regret}
J(\pi_{\lambda}) - J^{\star} = \sum_{t=0}^{\nt-1}\left\|\psi_{t}\right\|_H^2,
\end{align}
where $H\coloneqq B(R+B^\top P B)^{-1} B^\top$, and $\psi_{\ell,\nt}$ is defined as
\begin{align}
    \psi_{t}(\lambda)\coloneqq \sum_{\tau=t}^{\overline{t}}\left(F^\top\right)^{\tau-t} P \left(\lambda\phi_{\tau|t}-\phi_\tau^{\star}\right) - \sum^{T-1}_{\tau = \overline{t}+1}\left(F^\top\right)^{\tau-t} P\phi_\tau^{\star}.\nonumber
\end{align}
Thus,~\Cref{eq:overall_regret} leads to
\begin{align}
\nonumber
J(\pi_{\lambda}) - J^{\star} \leq &2\|H\|  \sum_{t=0}^{\nt-1} \left\|\sum_{\tau=t}^{\overline{t}}\left(F^{\top}\right)^{\tau-t}P\left(\lambda\phi_{\tau|t}-\phi_{\tau}^{\star}\right)\right\|^2\\
\label{eq:cost_gap}
&\quad +2\|H\|  \sum_{t=0}^{\nt-1} \left\|\sum_{\tau=\overline{t}+1}^{\nt-1}\left(F^{\top}\right)^{\tau-t}P\phi_{\tau}^{\star}\right\|^2.
\end{align}
Rearranging the terms and applying the Cauchy–Schwarz inequality, we obtain the following bound on the first term in~\eqref{eq:cost_gap}:
\begin{align}
\nonumber
& 2\|H\|\|P\|^2C_F^2\sum_{t=0}^{\nt-1}\left(\sum_{\tau=t}^{\overline{t}}\rho_F^{\tau-t}\left\|\lambda\phi_{\tau|t}-\phi_{\tau}^{\star}\right\|\right)^2\\
\label{eq:first_term_bound}
\leq & 2\|H\|\|P\|^2\frac{C_F^2}{1-\rho_F^2} \|\lambda\boldsymbol{\varepsilon}+(1-\lambda)\boldsymbol{\phi}^{\star}\|^2,
\end{align}
where $\boldsymbol{\phi}^{\star}$ stacks the terms $\phi_{\tau}^{\star}$ over $\tau\in [t,\overline{t}]$ and $t\in [\nt]$.
The second term in~\eqref{eq:cost_gap} can be bounded from above by
\begin{align*}
    2\gamma^2 C_F^2 \|H\|\|P\|^2\sum_{t=0}^{\nt-1} \left(\sum_{\tau=\overline{t}+1}^{\nt-1} \rho_F^{\tau-t} \right)^2\leq \frac{2\gamma^2 C_F^2 \|H\|\|P\|^2}{\left(1-\rho_F\right)^2} \rho_F^{2\nk}\nt.
\end{align*}
Denote by $\boldsymbol{\lambda}_t$, $\boldsymbol{\varepsilon}_t$ and  $\boldsymbol{1}$ vectors that stack ${\lambda}_t$, ${\varepsilon}_{t|\tau}$ and $1$ over $\tau\in [t,\overline{t}]$ where $(\lambda_t:t\in [\nt])$ are confidence parameters generated by \ouralg.
The \ouralg policy guarantees that
\begin{align*}
& 2\|H\|\|P\|^2\frac{C_F^2}{1-\rho_F^2}\sum_{t=0}^{\nt-1}\|\boldsymbol{\lambda}_t\cdot\boldsymbol{\varepsilon}_t+(\boldsymbol{1}-\boldsymbol{\lambda}_t)\cdot\boldsymbol{\phi}_t^{\star}\|^2\\
\leq & 2\|H\|\|P\|^2\frac{C_F^2}{1-\rho_F^2}\left(\varpi\left(\boldsymbol{\varepsilon},\boldsymbol{\phi}^{\star}\right) +  4\gamma^2\sqrt{\nt\nk^3+\nk^4}\right).
\end{align*}
We conclude that
\begin{align*}
\mathsf{CR}(\pi_{\ouralg})\leq  1&+\frac{1}{J^{\star}}\frac{2C_F^2 \|H\|\|P\|^2}{\left(1-\rho_F\right)^2}\Big( 
\gamma^2\rho_{F}^{2\nk}\nt+ 4\gamma^2\sqrt{\nt\nk^3+\nk^4}\\
& + \varpi\left(\boldsymbol{\varepsilon},\boldsymbol{\phi}^{\star}\right)
\Big),
\end{align*}
where we have applied~\Cref{lemma:dcl} to minimize over $\lambda\in\mathcal{I}$.
% The \textit{$\rho$-weighted inner product} $\langle \boldsymbol{a},\boldsymbol{b}\rangle_{\rho}$ between two concatenated vectors $\boldsymbol{a}=(a_{\tau|t})_{\tau\in [t,\overline{t}],t\in [\nt]}$ and $\boldsymbol{b}=(b_{\tau|t})_{\tau\in [t,\overline{t}],t\in [\nt]}$ with respect to a mapping $\rho:\mathbb{N}\rightarrow\mathbb{R}_+$ is defined as 
% \begin{align}
%     \label{eq:weighted_inner_product}
% \langle \boldsymbol{a},\boldsymbol{b}\rangle_{\rho} \coloneqq \sum_{t=1}^{\nt-1}\sum_{\tau=t}^{\overline{t}} \rho^2(\tau-t) a_{\tau|t}^{\top} b_{\tau|t}.
% \end{align}

% Similarly, based on~\eqref{eq:weighted_inner_product} the $\rho$-\textit{weighted norm} for a concatenated vector $\boldsymbol{a}$ is
% \begin{align}
%     \label{eq:weighted_norm}
% \left\|\boldsymbol{a}\right\|_{\rho}^2 = \langle\boldsymbol{a},\boldsymbol{a}\rangle_{\rho} = \sum_{t=1}^{\nt-1}\sum_{\tau=t}^{\overline{t}} \rho^2(\tau-t) \left\|a_{\tau|t}\right\|^2.
% \end{align}

Finally, we can treat~\eqref{eq:first_term_bound} differently such that
\begin{align*}
    & 2\|H\|\|P\|^2C_F^2\sum_{t=0}^{\nt-1}\left(\sum_{\tau=t}^{\overline{t}}\rho_F^{\tau-t}\left\|\lambda\phi_{\tau|t}-\phi_{\tau}^{\star}\right\|\right)^2\\
\leq 
&4\|H\|\|P\|^2 \frac{C_F^2}{1-\rho_F^2}\Bigg(\lambda^2\sum_{t=0}^{\nt-1} \sum_{\tau=t}^{\overline{t}}\|\varepsilon_{\tau|t}\|^2 \\
&\qquad +(1-\lambda)^2\sum_{t=0}^{\nt-1}\Big(\sum_{\tau=t}^{\overline{t}}\rho_F^{\tau-t}\left\|\phi_{\tau}^{\star}\right\|\Big)^2\Bigg).
\end{align*}
Since $\mathtt{Ad}(\phi_{\star})=\omega(\sqrt{\nt\log \nt})$, $J^{\star}= \omega(\sqrt{\nt\log\nt})$.
Implementing \dcl to the RHS of the bound above and taking the minimum over $\lambda$, the competitive ratio satisfies
\begin{align*}
&\mathsf{CR}(\pi_{\ouralg})\leq  1 + o(1)+\frac{O(1)}{J^{\star}} \cdot \\
& \frac{\Big(\sum_{t=0}^{\nt-1}\sum_{\tau=t}^{\overline{t}}\left\|\varepsilon_{\tau|t}\right\|^2\Big)\Big(\sum_{t=0}^{\nt-1}\big(\sum_{\tau=t}^{\overline{t}}\rho_F^{\tau-t}\left\|\phi_{\tau}^{\star}\right\|\big)^2\Big)}{\sum_{t=0}^{\nt-1}\Big(\sum_{\tau=t}^{\overline{t}}\left\|\varepsilon_{\tau|t}\right\|^2+\big(\sum_{\tau=t}^{\overline{t}}\rho_F^{\tau-t}\left\|\phi_{\tau}^{\star}\right\|\big)^2\Big)} 
\end{align*}
by setting $\nk=\log_{1/\rho_F} \nt$. Note that  $\sum_{t=0}^{\nt-1}(\sum_{\tau=t}^{\overline{t}}\rho_F^{\tau-t}\left\|\phi_{\tau}^{\star}\right\|)^2 = \omega(\sqrt{\nt\log\nt})$.
Applying the lower bound on $J^{\star}$ in~\cite{li2022robustness} for the offline optimal policy that minimizes~\eqref{eq:cost_lqc} subject to~\eqref{eq:dynamics_lqc}:
\begin{align*}
J^{\star}\geq C_0\sum_{t=0}^{\nt-1} \left(\sum_{\tau=t}^{\nt-1}\rho_F^{\tau-t}\|\phi_{\tau}^{\star}\|\right)^2,  \text{ for a constant}
\end{align*}
$C_0\coloneqq \frac{(1-\rho_F)^2}{2}\min\{\underline{\lambda}(P),\frac{\underline{\lambda}(R)}{\|B\|},\frac{\underline{\lambda}(Q)}{\max\{2,\|A\|\}}\}$, we derive the bound that $\mathsf{CR}(\ouralg)\leq  1 + O\left(\frac{\|\boldsymbol{\varepsilon}\|}{\|\boldsymbol{\varepsilon}\|+\sqrt{\nt\log\nt}}\right) + o(1)$.
\end{proof}

% The following impossibility result indicates that the bound in Theorem~\ref{thm:lac_lqc} is tight. If $J^{\star}= ...$, then ... 
% \begin{align*}
%   1+ \Omega\left(\frac{\varpi\left(\boldsymbol{\varepsilon},\boldsymbol{\phi}^{\star}\right)}{J^{\star}}\right)\leq  \mathsf{CR}(\ouralg)\leq 1 + O\left(\frac{\varpi\left(\boldsymbol{\varepsilon},\boldsymbol{\phi}^{\star}\right)}{J^{\star}}\right)
% \end{align*}

The following impossibility result indicates that the bound in Theorem~\ref{thm:lac_lqc} is tight. Theorem~\ref{thm:lqc_lower_bound} provides a lower bound on the competitive ratio for \textit{any} online policy. If $B$ is of full-rank and $\mathtt{Ad}(\phi_{\star})=\omega(\sqrt{\nt\log \nt})$, such that the optimal cost dominates the residual terms (i.e., $J^{\star} = \omega(\sqrt{\nt\log \nt})$), then the upper bound from Theorem~\ref{thm:lac_lqc} and the lower bound from Theorem~\ref{thm:lqc_lower_bound} coincide up to constant factors. This establishes a tight characterization for the competitive ratio of our \texttt{LAC} algorithm in the LQC setting:
\begin{align*}
1+ \Omega\left(\frac{\varpi\left(\boldsymbol{\varepsilon},\boldsymbol{\phi}^{\star}\right)}{J^{\star}}\right)\leq \mathsf{CR}(\ouralg)\leq 1 + O\left(\frac{\varpi\left(\boldsymbol{\varepsilon},\boldsymbol{\phi}^{\star}\right)}{J^{\star}}\right)
\end{align*}
This proves that our analysis is not loose and that the \texttt{LAC} policy is optimal for this problem class, up to the constants hidden in the asymptotic notation.

\begin{theorem}
\label{thm:lqc_lower_bound}
For the LQC problem,
any policy $\pi$ must satisfy
    \begin{align*}
\mathsf{CR}(\pi)\geq 1 + \frac{1}{J^{\star}}\frac{\underline{\sigma}^2(B)\underline{\sigma}^2(P)}{\overline{\lambda}\left(R+B^{\top}PB\right)} \varpi\left(\boldsymbol{\varepsilon},\boldsymbol{\phi}^{\star}\right).
    \end{align*}
\end{theorem}
Consequently, if $B$ is of full-rank, then any policy $\pi$ has the competitive ratio lower bound $\mathsf{CR}(\pi)\geq 1+\Omega(\varpi\left(\boldsymbol{\varepsilon},\boldsymbol{\phi}^{\star}\right)/J^{\star})$.

\begin{proof}[Proof of~\Cref{thm:lqc_lower_bound}]

Starting from~\eqref{eq:overall_regret}, we get $J(\pi_{\lambda}) - J^{\star} \geq \underline{\lambda}(H)\sum_{t=0}^{\nt-1}\left\|\psi_t\right\|^2$. Note that $\psi_t$ can be written as
\begin{align}
    \label{eq:psi_compact}
\psi_t = \sum_{\tau=t}^{\nt-1} \left(F^{\top}\right)^{\tau-t}P e_{\tau|t} = \mathcal{M}_{t} \mathbf{e}_{t},
\end{align}
where $\mathcal{M}_{t}$ is a block matrix $[P \ \ F^{\top}P \ \cdots \ (F^{\top})^{\nt-1-t}P ]$ and $\mathbf{e}_{t}$ is a stacked vector of errors $[e_{t|t} \ \ e_{t+1|t} \ \cdots e_{\nt-1|t}]^{\top}$ with
\begin{align*}
 e_{\tau|t}\coloneqq\begin{cases}
      \lambda\phi_{\tau|t}-\phi_\tau^{\star} &\text{ if } t\leq \tau\leq\overline{t} \\
      -\phi_\tau^{\star} &\text{ if } \overline{t}<\tau<\nt
 \end{cases}.
\end{align*}
Now, we can lower-bound the norm of $\psi_t$:
\begin{align*}
   \| \psi_{t}\|^2 = \|\mathcal{M}_{t} \mathbf{e}_{t}\|^2\geq \underline{\sigma}^2(\mathcal{M}_{t})\|\mathbf{e}_{t}\|^2\geq \underline{\sigma}^2(P)\|\mathbf{e}_{t}\|^2.
\end{align*}
Furthermore, truncating each $\mathbf{e}_{t}$ up to $\tau=\overline{t}$,
\begin{align*}
\sum_{t=0}^{\nt-1}\|\mathbf{e}_{t}\|^2\geq & \sum_{t=0}^{\nt-1}\sum_{\tau=t}^{\overline{t}}\left\|\lambda\phi_{\tau|t}-\phi_\tau^{\star}\right\|^2\\
\geq &
\min_{\lambda\in \mathcal{I}}\left\|\lambda\boldsymbol{\varepsilon}- (1-\lambda)\boldsymbol{\phi}^{\star}\right\|^2,
% \min_{\lambda\in \mathcal{I}}\sum_{t=0}^{\nt-1}\sum_{\tau=t}^{\overline{t}}\left\|\lambda\phi_{\tau|t}-\phi_\tau^{\star}\right\|^2.
\end{align*}
where $\boldsymbol{\varepsilon}$ and $\boldsymbol{\phi}^{\star}$ are stacked vectors defined in~\Cref{sec:problem_formulation}.
Solving this least-squares problem, we conclude that
\begin{align*}
J(\pi_{\lambda}) - J^{\star} \geq \underline{\lambda}(H)\underline{\sigma}^2(P)\varpi\left(\boldsymbol{\varepsilon},\boldsymbol{\phi}^{\star}\right).
\end{align*}

Finally, note that $\underline{\lambda}(H)= \min_{\|x\|=1}y^{\top} (R+B^{\top}PB)^{-1}y$ where $y\coloneqq B^{\top}x$. We know that for a positive definite matrix $M$, the quadratic form satisfies $y^{\top} M^{-1} y\geq\underline{\lambda}(M^{-1})\|y\|^2=\|y\|^2/\overline{\lambda}(M)$. Therefore, the proof completes by noting
\begin{align*}
\underline{\lambda}(H)=\min_{\|x\|=1}x^{\top} H x\geq \frac{\left\|B^{\top}x\right\|^2}{\overline{\lambda}\left(R+B^{\top}PB\right)}\geq \frac{\underline{\sigma}^2(B)}{\overline{\lambda}\left(R+B^{\top}PB\right)}.
\end{align*}
\end{proof}

\section{Practical Example: }
\label{sec:experiments}

To demonstrate the practical effectiveness of the proposed \ouralg algorithm, we conduct experiments on both a standard linear quadratic control (LQC) system and a $1$-dimensional nonlinear robotic arm model. All code and data to reproduce these experiments are publicly available\footnote{\url{https://github.com/tongxin-li/LAC}}. We set $\nt=200$ for all experiments, with a prediction horizon of $\nk=5$ for \texttt{MPC}, \texttt{Self-Tuning}~\cite{li2022robustness}, and a delay of $5$ with a learning rate $\beta=0.05$ for the \dcl procedure in \ouralg.

\subsection{Linear Quadratic Control}
\label{sec:linear_settings}

\begin{figure}[t]
    \centering    \includegraphics[width=1\linewidth]{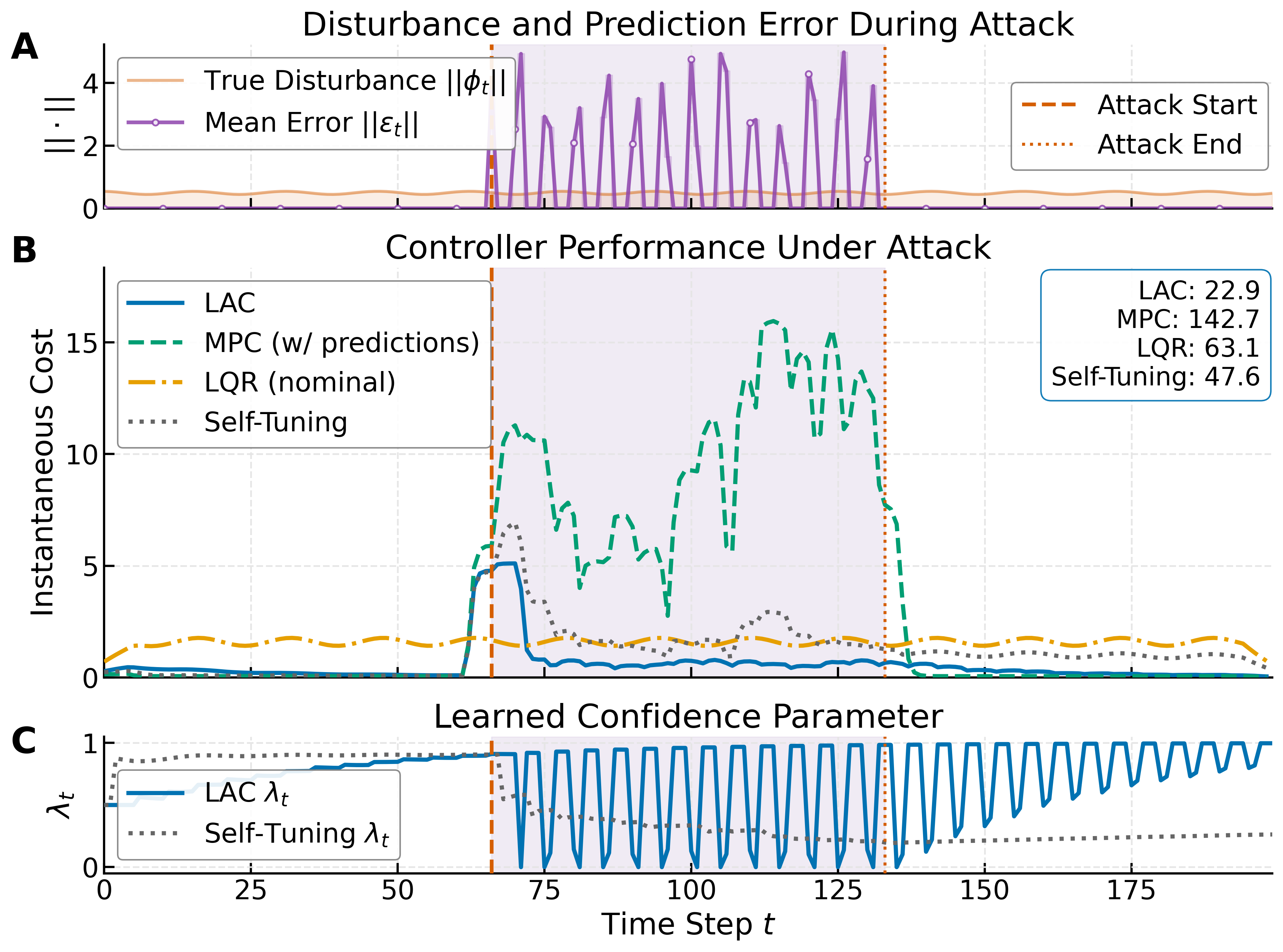}
    \caption{Linear trajectory tracking under adversarial prediction errors. \textsc{Top (\textbf{A}):} Prediction error during attack. 
     \textsc{Middle (\textbf{B}):} Instantaneous cost for each controller. 
        \textsc{Bottom (\textbf{C}):} Evolution of learned confidence parameter $\lambda_t$. Normalized total costs are also given on the right top of the subplot \textbf{B}.}
    \label{fig:attack}
\end{figure}

We first consider a constrained LQC problem to benchmark the robustness and adaptability of \ouralg against three baselines: Model Predictive Control with predictive parameters (\texttt{MPC}), Linear Quadratic Regulator (\texttt{LQR}) as a special case of MPC by setting nominal predictions of $(\phi_t^{\star}:t\in [\nt])$ as zeros, and a \texttt{Self-Tuning} controller~\cite{li2022robustness}. This problem involves controlling a robot whose movement is modeled as a discrete-time linear system with state $x_t = [p_t^{\top},v_t^{\top}]^{\top}\in\mathbb{R}^4$, where $p_t$ is position and $v_t$ is velocity. The controller must follow a fixed but unknown hypotrochoid trajectory $(y_t:t\in [\nt])$, defined by:
$$ y_t\coloneqq \begin{bmatrix} \cos(t/10)/2 + \cos(t/2)\\ \sin(t/10)/2 + \sin(t/2) \end{bmatrix}, \ \ t\in [\nt].$$
The robot's physical state consists of its position $p_t$ and velocity $v_t$, which evolve according to $p_{t+1}=p_t+ c_{1} v_t$ and $v_{t+1}=v_t+ c_{1} u_t$. To frame this as a tracking problem, we define the state as the tracking error $z_t \coloneqq p_t - y_t$, and the velocity $v_t$. This recasts the system into the canonical LQC state-space model  $x_{t+1} = Ax_t + Bu_t + \phi_t^{\star}$ where $x_t = [z_t^T, v_t^T]^T$ and:$$ A\coloneqq \begin{bmatrix} I_{2\times 2} & c_{1} \cdot I_{2\times 2}\\ 0_{2\times 2} & I_{2\times 2} \end{bmatrix}, \ B\coloneqq \begin{bmatrix} 0_{2\times 2} \\ c_{1} \cdot I_{2\times 2} \end{bmatrix}, \text{ and } \phi_t^{\star} \coloneqq Ay_t- y_{t+1}.$$ Furthermore, we impose a constraint $u_t\in [-u_{\max},u_{\max}]$ with $u_{\max}=10$ for all $t\in [\nt]$.
To minimize the tracking error while conserving energy, the quadratic cost is defined by the matrices $Q=\mathrm{diag}(1,1,0,0)$ and $R=I_{2\times 2}$.

\subsubsection{Cost vs. Prediction Error}
\label{sec:cost_vs_error}

To evaluate each controller's robustness against imperfect model knowledge, we designed an experiment to measure control cost as a function of prediction error. Prediction errors are introduced by adding a synthetic noise vector to the ground-truth $\nk$-step disturbance sequences. We set $c_{1}=0.2$. The synthetic noise vector is generated by first sampling a vector  from a multivariate Gaussian distribution $\mathcal{N}(I_{\nk\times\nk},\sigma^2 I_{\nk\times\nk})$ with $\sigma=0.5$ and the magnitude of this error is precisely controlled by normalizing and systematically varying its $\ell_2$-norm from 0 to 5 with a step size $0.1$. The mean error among $\nk$-step predictions is denoted by $\varepsilon_t$. 

Figure~\ref{fig:lqc} in~\Cref{sec:introduction} presents the control cost as a function of prediction error for all controllers. The \ouralg controller consistently achieves lower costs as prediction error increases, highlighting its ability to adaptively adjust confidence in the predictions. In contrast, the performance of classic \texttt{MPC} degrades rapidly with increasing error.  To ensure statistical significance, the total accumulated cost is averaged over 5 independent runs for each controller at every error level, with shadow regions displaying the range of the costs' standard deviation.

\subsubsection{Resilience to Adversarial Attacks}
\label{sec:exp_attack_linear}

To evaluate the controllers' resilience to fast fluctuating noise, we conduct an experiment where predictions are subjected to a temporary adversarial attack. We set $c_{1}=1$. For the first third of the simulation, disturbance predictions are perfect. During the middle third ($\nt/3 \leq t < 2\nt/3$), we simulate an intermittent attack by adding an adversarial error vector with a large $\ell_2$-norm $\|\varepsilon_t\|=4$ to the predictions (generated the same as in~\Cref{sec:cost_vs_error} with $\sigma=0.5$), whenever $(t \ \mathrm{ mod } \ 5)\in \{0,1\}$. In the final third, the attack ceases.

\begin{figure}[t]
    \centering    \includegraphics[width=1\linewidth]{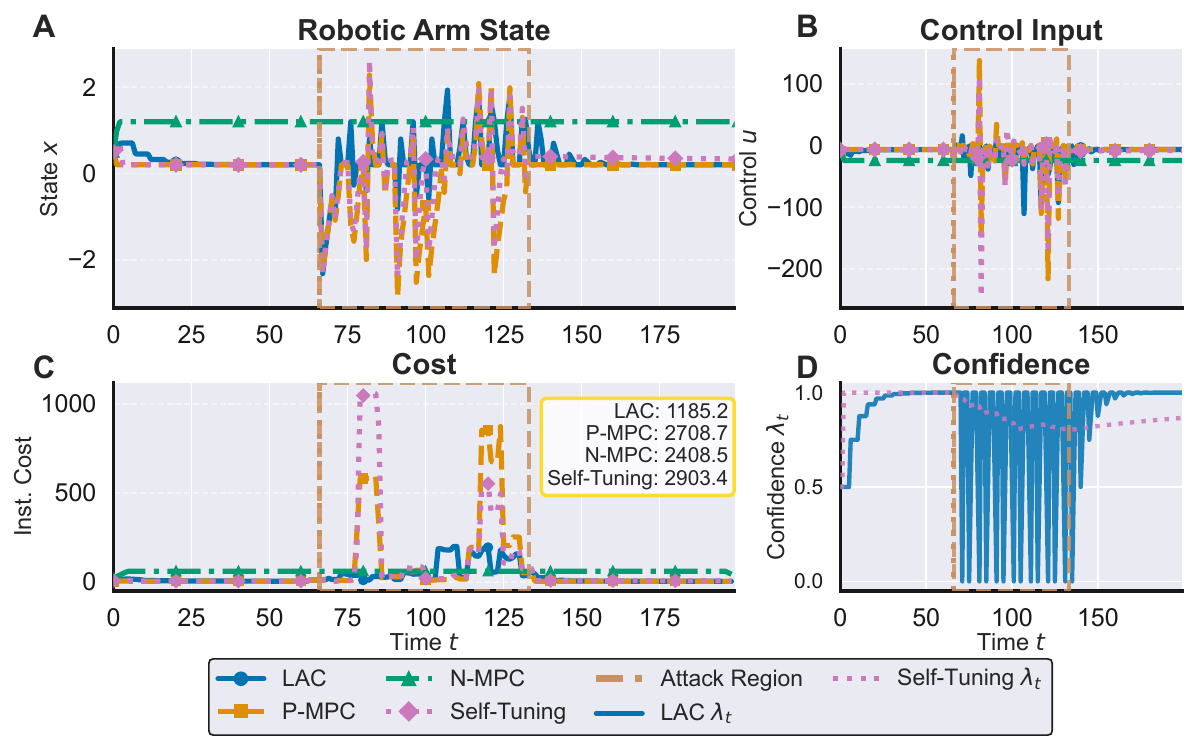}
    \caption{Performance and internal states of the controllers on the nonlinear robotic arm system under an adversarial prediction attack. (\textsc{\textbf{A}}) The state trajectory of the robotic arm over time. (\textsc{\textbf{B}}) The control input applied by each controller. (\textsc{\textbf{C}}) The instantaneous cost, with total accumulated costs annotated. (\textsc{\textbf{D}}) The learned confidence parameter $\lambda_t$ for two adaptive controllers, \ouralg and \texttt{Self-Tuning}. }
    \label{fig:robot_arms}
\end{figure}

The experimental results demonstrate \ouralg's superior robustness and adaptability. As shown in Figure~\ref{fig:attack}B, while the standard \texttt{MPC} controller becomes unstable and incurs catastrophic costs by naively trusting the corrupted predictions, \ouralg's performance remains stable and close to that of the conservative \texttt{LQR} baseline. The mechanism behind this resilience is revealed in Figure~\ref{fig:attack}C. I.e., upon detecting the attack, \ouralg's confidence parameter, 
$\lambda_t$, rapidly drops to nearly zero, effectively ignoring the adversarial predictions. By learning to distrust its compromised model, \ouralg achieves the lowest overall cost ($22.9$; normalized by $5$), thus verifying its ability to securely operate in unpredictable environments.

\subsection{Robot Arm Control}
\label{sec:robot_arm_model}

We further evaluate the controllers on a nonlinear robotic arm model. The system's dynamics function $f_t$ is given by:
\begin{align}
\label{eq:robotic_arm_dynamics}
    x_{t+1} = x_t + c_2 \sin(x_t) + c_3 u_t \exp(-|x_t|) + \phi_t^\star
\end{align}
with parameters $c_2=0.5$ and $c_3=0.2$. The state and input constraint sets are:
\begin{align}
\nonumber
    \mathsf{X} = \{ x \in \mathbb{R} : |x| \le 0.2 \}, \quad \mathsf{U} = \{ u \in \mathbb{R} : |u| \le 10^3 \}.
\end{align}
The cost for this system is defined as $c(x_t, u_t) = x_t^2 + c_4 u_t^2$, with $c_4=0.1$.
The system is subject to adversarial attacks on the disturbance predictions, generated the same as in \Cref{sec:exp_attack_linear} with a large $\ell_2$-norm $\|\varepsilon_t\|=4$ during a specified time window, as visualized in Figure~\ref{fig:robot_arms}. We compare \ouralg with the MPC with predictive parameters (\texttt{P-MPC}), MPC with nominal zero parameters (\texttt{N-MPC}), and
\texttt{Self-Tuning}, which is implemented using a linearized system $x_{t+1} = (1+c_2)x_t + c_3 u_t$ of~\eqref{eq:robotic_arm_dynamics} since it only applies to linear systems.

Figure~\ref{fig:robot_arms} provides a detailed breakdown of the state, control input, cost distribution, and confidence evolution for the robotic arm before, under, and after the attack. The \ouralg controller demonstrates superior resilience, with lower cost spikes and more stable confidence adaptation during the attack window.

\section{Conclusion}

In this work, we analyze Learning-Augmented Control (\ouralg), a framework that  integrates untrusted ML predictions into constrained, nonlinear MPC. By learning a confidence parameter $\lambda_t$ online via the Delayed Confidence Learning (\dcl) procedure, \ouralg adaptively balances between uncertain predictions and a safe nominal policy. We provide formal competitive ratio guarantees that prove this ``best-of-both-worlds'' behavior, a result supported by experiments showing that \ouralg remains robust and outperforms baselines like \texttt{MPC} and \texttt{LQR} under both graded prediction errors and direct fast fluctuating adversarial attacks.

Our framework, while robust, has several limitations that suggest avenues for future research. The current formulation relies on receiving predictions of model parameters and observing their true values after a delay. A key direction is to extend \texttt{LAC} to handle different information structures, such as learning from direct action advice provided by an external policy. Furthermore, our analysis assumes perfect observability post-delay. In future work, it would be interesting to address partial or noisy feedback by integrating techniques from system identification or state estimation into the confidence learning loop. Another promising avenue is to explore more advanced online learning schemes beyond online convex optimization to handle scenarios where the surrogate loss is non-convex. Finally, a key theoretical question left open is the tightness of our competitive ratio bound for general nonlinear systems. While this work establishes matching upper and lower bounds for the LQC case (Theorem~\ref{thm:lac_lqc} and \ref{thm:lqc_lower_bound}), thereby tightly characterizing its fundamental performance limits, our main result for the general setting (Theorem~\ref{thm:lac}) only provides an upper bound. A significant direction for future research is to establish a corresponding impossibility result, or a lower bound on the competitive ratio, for the general nonlinear case.

% \addcontentsline{toc}{section}{Bibliography}
\bibliographystyle{ieeetr}
{\bibliography{main}}

\appendix

\subsection{Proof of Lemma~\ref{lemma:regret_error}}
\label{app:proof_lemma_regret_error}

% We define $\xi_{\tau | t}(\lambda)\coloneqq \rho(\tau-t) \big(\lambda\varepsilon_{\tau |t}
% + (1-\lambda)\overline{\varepsilon}_{\tau |t}\big)$ and $\xi_{t,\nt}(\lambda)\coloneqq \sum_{\tau=t}^{\overline{t}}\xi_{\tau | t}(\lambda)$ 
%  for notational convenience.
% Let $\overline{\eta}_t(x_t)\coloneqq u_t(x_t;\kappa_{t:\overline{t}|t})  - u_t(x_t;\phi_{t:\nt}^{\star})$ and denote $e_t^{\lambda}(x_t)\coloneqq \lambda\eta_t(x_t)+(1-\lambda)\overline{\eta}_t(x_t)$ given a confidence parameter $\lambda\in\mathcal{I}$. 

The \textit{per-step action error} $e^{u}_t(x_t)$ at time $t\in [\nt]$ defined below characterizes the discrepancy between the MPC solution $u_t$ and the clairvoyant optimal action $u_t(x_t;\phi_{t:\nt}^{\star})$, both given the same system state $x_t$ generated by a given controller (in our context, it is the MPC scheme in~\eqref{eq:mpc1}-\eqref{eq:mpc2}),
\begin{align} 
\label{eq:per_step_error}
    e^{u}_t(x_t;\phi_{t:\overline{t}|t})\coloneqq \left\|u_t\left(x_t;\phi_{t:\overline{t}|t}\right) - u_t\left(x_t;\phi_{t:\nt}^{\star}\right)\right\|.
\end{align}
In future contexts, we use the shorthand $e_t^{u}(x_t)$ for $ e^{u}_t(x_t;\lambda \phi_{t:\overline{t}}+(1-\lambda) \kappa_{t:\overline{t}})$ if there is no ambiguity. Similarly, we define the \textit{per-step state error} $e_t^{x}(x_{t-1})$ as
\begin{align} 
\nonumber
    e^{x}_t(x_{t-1};\phi_{t:\overline{t}|t})\coloneqq \left\|x_{t}\left(x_{t-1};\phi_{t:\overline{t}|t}\right) - x_{t}\left(x_{t-1};\phi_{t:\nt}^{\star}\right)\right\|.
\end{align}

\begin{proof}[Proof of~\Cref{lemma:regret_error}]

Recall that $x_t(x_{\tau};\phi_{\tau:\nt}^{\star})$ represents the corresponding optimal state at time $t\in [\nt]\cup \{\nt\}$ induced by applying clairvoyant optimal actions $u_{\tau:t}(x_\tau;\phi_{\tau:\nt}^{\star})$ with a fixed starting state $x_\tau$ at time $\tau\leq t$. If $\tau=0$ and $x_0$ is the initial state, we simplify it as $x_t^{\star}$. Furthermore, let $x_t^{\pi}$ (or $u_t^{\pi}$) be the corresponding state (or action) at time $t$ induced by applying $\pi_{\lambda}$ where $\pi_{\lambda}(x_t)=u_t\left(x_t;\lambda \phi_{t:\overline{t}}+(1-\lambda) \kappa_{t:\overline{t}}\right) \ \text{($\lambda$-confident control)}$. Note that $x_0^{\star}\equiv x_0^{\pi}\equiv x_0$. Rewriting the state difference $x_t^{\pi} - x_t^{\star}$ into a telescoping sum, it follows that for all $t\in[T]\cup\{T\}$,
\begin{align}
\nonumber
\left\|x_t^{\pi} - x_t^{\star}\right\| \leq  & \left\|x_t^{\pi}-x_{t}(x_{t-1}^{\pi};\phi_{t-1:\nt}^{\star})\right\| \\ 
\label{eq:performance_diff_1}
& +\sum_{\tau=1}^{t-1}\left\|x_{t}(x_{\tau}^{\pi};\phi_{\tau:\nt}^{\star})-x_{t}(x_{\tau-1}^{\pi};\phi_{\tau-1:\nt}^{\star})\right\|.
\end{align}
The decomposition in~\eqref{eq:performance_diff_1} is from the performance difference lemma~\cite{kakade2002approximately} and the perturbation analysis in~\cite{lin2022bounded}.
Using the EDPB property in~\eqref{eq:perturbation_bound_state} of~\Cref{def:perturbation_bound}, there exists a mapping $\rho:\mathbb{N}\rightarrow\mathbb{R}_+$ such that $\sum_{t\in [\nt]}\rho(t)\leq C$ and $\rho(0)\geq 1$ such that
\begin{align*}
&\sum_{\tau=1}^{t-1}\left\|x_{t}(x_{\tau}^{\pi};\phi_{\tau:\nt}^{\star})-x_{t}(x_{\tau-1}^{\pi};\phi_{\tau-1:\nt}^{\star})\right\|\\
\leq &\sum_{\tau=1}^{t-1}\rho(t-\tau)\left\|x_{\tau}^{\pi}-x_{\tau}(x_{\tau-1}^{\pi};\phi_{\tau-1:\nt}^{\star})\right\|,
\end{align*}
which together with~\eqref{eq:performance_diff_1} give that for all $t\in[T]\cup\{T\}$,
\begin{align}
\label{eq:performance_diff_2}
    \left\|x_t^{\pi} - x_t^{\star}\right\| \leq\sum_{\tau=1}^{t}\rho(t-\tau)\left\|x_{\tau}^{\pi}-x_{\tau}(x_{\tau-1}^{\pi};\phi_{\tau-1:\nt}^{\star})\right\|.
\end{align}
% Following the system update rule~\eqref{eq:dynamics},
% \begin{align*}
%     x_{\tau}^{\pi} &=f_{\tau-1}\left(x_{\tau-1}^{\pi},u_{\tau-1}^{\pi};\phi_{\tau-1}^{\star}\right),\\
%     x_{\tau}(x_{\tau-1}^{\pi};\phi_{\tau-1:\nt}^{\star}) &= f_{\tau-1}\left(x_{\tau-1}^{\pi},u_{\tau-1}(x_{\tau-1}^{\pi};\phi_{\tau-1:\nt}^{\star});\phi_{\tau-1}^{\star}\right).
% \end{align*}
% Since $f_t$ is $L_f$-Lipschitz for all $t\in[\nt]$, 
Using the definition in~\eqref{eq:per_step_error}, as a result, 
\begin{align*}
\left\|x_{t}^{\pi}-x_{t}(x_{t-1}^{\pi};\phi_{t-1:\nt}^{\star})\right\|
= & e_{t}^{x}(x_{t-1}^{\pi}).
\end{align*}
% where we have applied the control policy definition  $u_{t}^{\pi}=\lambda u_{t}(x_{t}^{\pi};\phi_{t:\overline{t}|t})+(1-\lambda)u_{t}(x_{t}^{\pi};\kappa_{t:\overline{t}|t})$ for all $t\in [\nt]$. 
Thus, the LHS of~\eqref{eq:performance_diff_2} becomes
\begin{align}
\label{eq:performance_diff_x}
     \left\|x_t^{\pi} - x_t^{\star}\right\| \leq \sum_{\tau=1}^{t}\rho(t-\tau) e_{\tau}^{x}(x_{\tau-1}^{\pi}).
\end{align}
Similar to~\eqref{eq:performance_diff_1}, we decompose the action difference $u_t^{\pi} - u_t^{\star}$ as
\begin{align}
\nonumber
\left\|u_t^{\pi} - u_t^{\star}\right\| \leq  & \left\|u_t^{\pi}-u_{t}(x_{t}^{\pi};\phi_{t:\nt}^{\star})\right\| \\ 
\label{eq:performance_diff_3}
& +\sum_{\tau=1}^{t}\left\|u_{t}(x_{\tau}^{\pi};\phi_{\tau:\nt}^{\star})-u_{t}(x_{\tau-1}^{\pi};\phi_{\tau-1:\nt}^{\star})\right\|.
\end{align}
Applying the EDPB property in~\eqref{eq:perturbation_bound_action} of~\Cref{def:perturbation_bound} implies
\begin{align}
\label{eq:performance_diff_4}
    \left\|u_t^{\pi} - u_t^{\star}\right\| \leq e_t^{u}(x_{t}^{\pi})+\sum_{\tau=1}^{t}\rho(t-\tau)e_{\tau}^{x}(x_{\tau-1}^{\pi}).
\end{align}

Continuing from~\eqref{eq:performance_diff_x}, by the Cauchy-Schwarz inequality,
\begin{align}
\nonumber
\left\|x_t^{\pi} - x_t^{\star}\right\|^2 \leq & \left(\sum_{\tau=1}^{t}\rho(t-\tau) e_{\tau}^{x}(x_{\tau-1}^{\pi})\right)^2\\
\label{eq:performance_diff_x_2}
\leq & C \sum_{\tau=1}^{t}\rho(t-\tau) \left(e_{\tau}^{x}(x_{\tau-1}^{\pi})\right)^2.
\end{align}
Similarly for~\eqref{eq:performance_diff_4},
\begin{align}
\nonumber
\left\|u_t^{\pi} - u_t^{\star}\right\|^2 \leq & \left(e_t^{u}(x_{t}^{\pi})+\sum_{\tau=1}^{t}\rho(t-\tau)e_{\tau}^{x}(x_{\tau-1}^{\pi})\right)^2\\
\label{eq:performance_diff_u_2}
\leq & 2\left(e_t^{u}(x_{t}^{\pi})\right)^2 + 2C \sum_{\tau=1}^{t}\rho(t-\tau) \left(e_{\tau}^{x}(x_{\tau-1}^{\pi})\right)^2.
\end{align}
% Since the costs are $L_c$-Lipschitz, using the bounds in~\eqref{eq:performance_diff_x} and~\eqref{eq:performance_diff_3}, 

Since the costs are non-negative, convex, and $\ell$-smooth, Lemma F.2 in~\cite{lin2021perturbation} implies that for any $\eta>0$, we can bound the cost difference $J(\pi_{\lambda}) - (1+\eta)J^{\star}$ from above by
\begin{align}
\nonumber
\frac{\ell}{2}\left(1+\frac{1}{\eta}\right)\left(\sum_{t=1}^{\nt}\left\|x_t^{\pi}-x_t^{\star}\right\|^2 + \sum_{t=0}^{\nt-1}\left\|u_t^{\pi}-u_t^{\star}\right\|^2\right).
\end{align}
Applying~\eqref{eq:performance_diff_x_2} and~\eqref{eq:performance_diff_u_2}, the RHS of above is bounded by
\begin{align*}
&\frac{\ell}{2}\left(1+\frac{1}{\eta}\right)\sum_{t=0}^{\nt}\left(2\left(e_t^{u}(x_{t}^{\pi})\right)^2+3C \sum_{\tau=1}^{t}\rho(t-\tau) \left(e_{\tau}^{x}(x_{\tau-1}^{\pi})\right)^2\right)\\
\leq &\frac{\ell}{2}\left(1+\frac{1}{\eta}\right)\left(2\sum_{t=0}^{\nt-1}\left(e_t^{u}(x_{t}^{\pi})\right)^2+3C^2\sum_{t=1}^{\nt}\left(e_t^{x}(x_{t-1}^{\pi})\right)^2\right)
\end{align*}
by letting $e_{\nt}^{u}(x_{\nt}^{\pi})\equiv 0$. 
% \begin{align}
% \nonumber
%     J(\pi_{\lambda}) - J^{\star} \leq & 2L_c L_f\sum_{t=1}^{\nt-1} \left(\sum_{\tau=1}^{t} \rho(t-\tau) e_{\tau-1}^{\lambda}(x_{\tau-1}^{\pi})\right)\\
% \nonumber
% & + L_c L_f\sum_{\tau=1}^{\nt} \rho(t-\tau) e_{\tau-1}^{\lambda}(x_{\tau-1}^{\pi})\\
% \label{eq:cost_gap}
% & + L_c\sum_{t=0}^{\nt-1} e_t^{\lambda}(x_{t}^{\pi}).
% \end{align}
Using the sensitivity bounds in~\eqref{eq:sensitivity_bound_action} and~\eqref{eq:sensitivity_bound_state}, both $\sum_{t=0}^{\nt-1} \left(e_t^{u}(x_t^{\pi})\right)^2$ and $\sum_{t=1}^{\nt} \left(e_t^{x}(x_{t-1}^{\pi})\right)^2$ satisfy the following upper bound:
\begin{align}
\nonumber
% \nonumber
% \sum_{t=0}^{\nt-1}\Bigg(\lambda &\Big(\sum_{\tau=t}^{\overline{t}} \rho(\tau-t)\varepsilon_{\tau |t}+\rho(\nk)\Big) + \\
% \label{eq:upper_bound_1}
% &(1-\lambda)\Big(\sum_{\tau=t}^{\overline{t}} \rho(\tau-t)\overline{\varepsilon}_{\tau |t} + \rho(\nk)\Big)\Bigg)^2.
&\sum_{t=0}^{\nt-1}\left(\sum_{\tau=t}^{\overline{t}}\rho(\tau-t)\left\|\lambda \phi_{\tau|t} +(1-\lambda)\kappa_{\tau|t} - \phi_{\tau}^{\star}\right\|+\gamma\rho(\nk)\right)^2\\
\label{eq:upper_bound_1}
\leq & \sum_{t=0}^{\nt-1}\left(\sum_{\tau=t}^{\overline{t}}\rho(\tau-t)\left\|\lambda \varepsilon_{\tau|t} +(1-\lambda)\overline{\varepsilon}_{\tau|t}\right\|+\gamma\rho(\nk)\right)^2 .
\end{align}
Rearranging the terms,~\eqref{eq:upper_bound_1} satisfies
\begin{align}
\nonumber
\frac{\eqref{eq:upper_bound_1}}{2}\leq & \nt\gamma^2\rho^2(\nk) + \sum_{t=0}^{\nt-1}\Big(\sum_{\tau=t}^{\overline{t}} \rho(\tau-t)\left\|\lambda \varepsilon_{\tau|t} +(1-\lambda)\overline{\varepsilon}_{\tau|t}\right\|  \Big)^2.
\end{align}
% For the second term in~\eqref{eq:upper_bound_2}, applying the Cauchy–Schwarz inequality,
% \begin{align*}
% & \sum_{t=0}^{\nt-1}\Big(\sum_{\tau=t}^{\overline{t}} \rho(\tau-t)\left\|\lambda \varepsilon_{\tau|t} +(1-\lambda)\overline{\varepsilon}_{\tau|t}\right\| \Big)^2\\
% \leq & C^2\sum_{t=0}^{\nt-1}\sum_{\tau=t}^{\overline{t}}\left\|\lambda \varepsilon_{\tau|t} +(1-\lambda)\overline{\varepsilon}_{\tau|t}\right\|^2 = C^2\left\|\lambda\boldsymbol{\varepsilon}+(1-\lambda)\overline{\boldsymbol{\varepsilon}}\right\|^2.
% \end{align*}
% Hence, it follows that
% \begin{align}
% \nonumber
% \sum_{t=0}^{\nt-1} \left(e_t^{u}(x_t^{\pi})\right)^2 &\leq  2\gamma^2\rho^2(\nk)\nt + 
% 2C^2\left\|\lambda\boldsymbol{\varepsilon}+(1-\lambda)\overline{\boldsymbol{\varepsilon}}\right\|^2,\\
% \nonumber
% \sum_{t=1}^{\nt} \left(e_t^{x}(x_{t-1}^{\pi})\right)^2 &\leq  2\gamma^2\rho^2(\nk)\nt + 
% 2C^2\left\|\lambda\boldsymbol{\varepsilon}+(1-\lambda)\overline{\boldsymbol{\varepsilon}}\right\|^2.
% \end{align}
Hence, it follows that
\begin{align}
\nonumber
\sum_{t=0}^{\nt-1} \left(e_t^{u}(x_t^{\pi})\right)^2 &\leq  2\gamma^2\rho^2(\nk)\nt + 
2\xi(\lambda),\\
\nonumber
\sum_{t=1}^{\nt} \left(e_t^{x}(x_{t-1}^{\pi})\right)^2 &\leq  2\gamma^2\rho^2(\nk)\nt + 
2\xi(\lambda).
\end{align}
Consequently, for any $\eta>0$, $J(\pi_{\lambda}) - (1+\eta)J^{\star}$ can be bounded from above by
\begin{align}
\nonumber
  &  \frac{\ell}{2}\left(1+\frac{1}{\eta}\right)\left(3C^2+2\right)\sum_{t\in [\nt]} \left(e_t^{u}(x_t^{\pi})\right)^2  \leq  \frac{(1+\eta)(3C^2+2)}{\eta}\ell\\
\label{eq:upper_bound_noncanonical}
& \quad \cdot \left(  \gamma^2\rho^2(\nk)\nt +
\xi(\lambda)\right).
% & \quad \cdot \left(  \gamma^2\rho^2(\nk)\nt +
% C^2\left\|\lambda\boldsymbol{\varepsilon}+(1-\lambda)\overline{\boldsymbol{\varepsilon}}\right\|^2\right).
\end{align}

\end{proof}

\subsection{Proof of Lemma~\ref{lemma:dcl}}
\label{app:proof_lemma_dcl_regret}

For notional simplicity, we omit $\nt$ in $\xi_{t,\nt}(\lambda)$ and write it as $\xi_{t}(\lambda)$.
First, relabeling the time indices, the regret of \dcl can be represented as 
\begin{align*}
\mathsf{Regret}(\dcl) \coloneqq & \sum_{t=0}^{\nt-1}\xi_{t}\left(\lambda_t\right) -  \min_{\lambda\in\mathcal{I}}\sum_{t=0}^{\nt-1}\xi_{t}\left(\lambda\right)\\
= & \sum_{j=0}^{\nk-1}\sum_{t\in \mathcal{\nt}_j}\xi_{t}\left(\lambda_t\right) -  \min_{\lambda\in\mathcal{I}}\sum_{j=0}^{\nk-1}\sum_{t\in \mathcal{\nt}_j}\xi_{t}\left(\lambda\right),
\end{align*}
where $\mathcal{\nt}_j\coloneqq\left\{j+\nk i<\nt: i\in\mathbb{N}\right\}$ satisfies $|\mathcal{\nt}_j|\leq {\nt}/{\nk}+1$. 

\begin{proof}[Proof of~\Cref{lemma:dcl}]
The delayed 
update rule in Algorithm~\ref{alg:dcl} is equivalent to the implementation of $\nk$ subroutines on $\mathcal{\nt}_j$ for $j=0,\ldots,\nk-1$, with each $j$-th subroutine runs the following update rule:
\begin{align*}
\lambda_i^{(j)} = \mathsf{Proj}_{\mathcal{I}}\left(\lambda_{i-1}^{(j)} + {\beta}\nabla_{\lambda}\xi_{i-1}^{(j)}(\lambda)\right), \ i\in\mathbb{N}
\end{align*}
so that $\lambda_i^{(j)}\equiv\lambda_{j+\nk i}$ ($j+\nk i<\nt$).
Therefore, applying the analysis of black-box online learning algorithms in
\cite{weinberger2002delayed,joulani2013online}, the regret of \dcl can be bounded by a summation of $k$ regret corresponding to $\nk$ subroutines such that
\begin{align}
\label{eq:dcl_bound_1}
\mathsf{Regret}(\dcl)\leq \sum_{j=0}^{\nk-1}\left(\sum_{t\in \mathcal{\nt}_j}\xi_{t}\left(\lambda_t\right)-\min_{\lambda\in\mathcal{I}}\sum_{t\in \mathcal{\nt}_j}\xi_{t}\left(\lambda\right)\right).
\end{align}

Now, we bound the regret of the subroutines.  By the definition of $\xi_{t}(\lambda)$, $\nabla_{\lambda}\xi_{t}(\lambda)$ satisfies
\begin{align}
\nonumber
  \left|\nabla_{\lambda}\xi_{t}(\lambda)\right|\leq  & 2\left(\sum_{\tau=t}^{\overline{t}}\rho(\tau-t)\left\|\lambda \varepsilon_{\tau|t} +(1-\lambda)\overline{\varepsilon}_{\tau|t}\right\|\right) \\
  \label{eq:proof_gradient}
  & \cdot \left(\sum_{\tau=t}^{\overline{t}}\rho(\tau-t)\|\varepsilon_{\tau|t}-\overline{\varepsilon}_{\tau|t}\|\right)\leq 2C^2\gamma^2,
\end{align}
using the Cauchy–Schwarz inequality.
% \begin{align}
% \nonumber
% \nabla_{\lambda}\xi_{t}(\lambda) \leq &\nabla_{\lambda}\sum_{\tau=t}^{\overline{t}}\left\|\lambda \varepsilon_{\tau|t} +(1-\lambda)\overline{\varepsilon}_{\tau|t}\right\|^2 
% =   \sum_{\tau=t}^{\overline{t}} \nabla_{\lambda}\left\|\lambda \varepsilon_{\tau|t} +(1-\lambda)\overline{\varepsilon}_{\tau|t}\right\|^2\\
% = & 2\sum_{\tau=t}^{\overline{t}}\left(\lambda\varepsilon_{\tau|t}+(1-\lambda)\overline{\varepsilon}_{\tau|t}\right)^{\top}\left(\varepsilon_{\tau|t}-\overline{\varepsilon}_{\tau|t}\right) \leq  2\nk\gamma^2
% \end{align}
% Therefore, $|\nabla_{\lambda}\xi_{t,\nt}(\lambda)|\leq 2C^2\gamma^2$.
Since $\xi_{t,\nt}(\lambda)$ is a convex function of $\lambda$ for all $t\in [\nt]$, the online mirror descent update rule in Algorithm~\ref{alg:dcl} is equivalent to the follow the regularized leader (FTRL) policy~\cite{hazan2022introduction} with an $\ell_2$-norm regularizer, leading to the following regret bound for each $j$-th subroutine
\begin{align}
\nonumber
\sum_{t\in \mathcal{\nt}_j}\xi_{t}\left(\lambda_t\right)-\min_{\lambda\in\mathcal{I}}\sum_{t\in \mathcal{\nt}_j}\xi_{t}\left(\lambda\right)\leq & \frac{1}{2\eta} + 2\eta\sum_{t\in\mathcal{\nt}_j}|\nabla_{\lambda}\xi_{t}(\lambda)|^2\\
\label{eq:dcl_bound_2}
\leq & 4C^2\gamma^2\sqrt{\nt/\nk+1}.
\end{align}

Therefore, the selection of $\lambda_t$ at each time $t\in [\nt]$ follows from a delayed online learning procedure \dcl, which can be regarded as the implementation of the $\nk$ online optimization subroutines. Combining~\eqref{eq:dcl_bound_1} and~\eqref{eq:dcl_bound_2},
\begin{align}
\mathsf{Regret}(\dcl)\leq 4C^2\gamma^2\sqrt{\nt\nk+\nk^2}. 
\end{align}
\end{proof}
\end{document}